\documentclass{IEEEtran}

\usepackage{balance}
\usepackage{amsmath,graphicx}
\usepackage{amsfonts}
\usepackage{color}
\usepackage{amssymb}
\usepackage{amsthm}
\usepackage{bm}
\usepackage{cite}

\newtheorem{lemma}{Lemma}
\newtheorem{corollary}{Corollary}
\newtheorem{definition}{Definition}
\newtheorem{remark}{Remark}
\usepackage{cuted}
\usepackage{array,multirow}
\usepackage{stfloats}
\usepackage{bbm}
\usepackage{mathtools}
\usepackage{algorithm,algorithmic} %
\usepackage{multicol}

\newcommand{\revisNew}{\textcolor{black}}
\newcommand{\arxivNew}{\textcolor{black}}

\begin{document}


\title{Arithmetic Average Density Fusion \revisNew{- Part II}: Unified Derivation for Unlabeled and Labeled RFS Fusion}

\author{Tiancheng~Li,~\IEEEmembership{Senior Member,~IEEE}%
\thanks{
Manuscript submitted to IEEE T-AES on 03-Dec-2022, 23-May-2023, revised on 15-Nov-2023} 
\thanks{This work was partially supported 
 by National Natural Science Foundation of China (Grant No. 62071389), Natural Science Basic Research Program of Shaanxi Province (Program No. 2023JC-XJ-22), 
 and the Fundamental Research Funds for the Central Universities. 
}
\thanks{T.\ Li is with the Key Laboratory of Information Fusion Technology (Ministry of Education), School of Automation, Northwestern Polytechnical University, Xi'an 710129, China, e-mail: t.c.li@nwpu.edu.cn
}
}

\maketitle

\begin{abstract}
\revisNew{As a fundamental information fusion approach, the arithmetic average (AA) fusion has recently been investigated for various random finite set (RFS) filter fusion in the context of multi-sensor multi-target tracking.
It is not a straightforward extension of the ordinary density-AA fusion to the RFS distribution but has to preserve the form of the fusing multi-target density.}
In this work, we first propose a statistical concept, probability hypothesis density (PHD) consistency, and explain how it can be achieved by the PHD-AA fusion and
lead to more accurate and robust detection and localization of the present targets. This forms a both theoretically sound and technically meaningful reason for performing \revisNew{inter-filter PHD AA-fusion/consensus, while preserving the form of the fusing RFS filter.}
Then, we derive and analyze the proper AA fusion formulations for most existing unlabeled/labeled RFS filters basing on the (labeled) PHD-AA/consistency. 
These derivations are \revisNew{theoretically unified,} exact, need no approximation and 
greatly enable heterogenous unlabeled and labeled RFS density fusion 
which is separately demonstrated in \arxivNew{two consequent companion papers}. 
\end{abstract}





%


\section{Introduction}
\revisNew{Multi-target tracking (MTT) that involves estimating both} the number {and} states of multiple targets in the presence of clutter and misdetection 
has been a long standing research topic with a wide range of applications including surveillance, traffic control and autonomous driving \cite{Vo15mtt}. 
Of particular significance in real world is the scenario involving a time-varying, random perhaps, number of targets that may appear and disappear from the field of view of the sensors at any time in which the number and even identities of targets need to be estimated jointly with the states, \revisNew{namely joint detection and tracking}. One of the most attractive approaches to this problem is the random finite set (RFS) theory \cite{Mahler14book} which has led to various probabilistic filters \revisNew{based on different scenario models and needs}, such as the probability hypothesis density (PHD) filter \cite{Mahler03,Vo06}, 
the cardinazied PHD (CPHD) filter \cite{Vo07cphd}, the multi-Bernoulli (MB) filter \cite{Vo09CBmember}, the Bernoulli filter \cite{Vo11Bernoulli,Ristic13tutorialBernoulli}, the Poisson MB mixture (PMBM) filter \cite{Williams15taesPMBM}, the generalized labeled MB (GLMB) filter \cite{Vo13Label,Vo14GLMB} and their variants.

\revisNew{The growing prevalence of wireless} sensor networks \cite{Olfati-Saber07,Sayed14book,Meyer18ProIEEE} has attracted an increasing interest in robust multi-sensor MTT methods that can take advantage of large quantities of gathered information from multiple sensors for improved estimation accuracy, extended tracking coverage, and enhanced viability/scalability of the local sensors. Existing multi-sensor MTT approaches can be classified into two main groups differing in the type of the information shared between sensors: data-level and estimate-level \cite{Da21Recent}. Briefly speaking, the 
majority of the former is dedicated to compute the exact multi-sensor likelihood/posterior. It usually uses a centralized sensor network where the measurement data of all sensors are sent to the fusion center for joint, optimal fusion. This, however, suffers from poor computational scalability \revisNew{(with respect to the number of sensors)} due to the intractable multi-sensor measurement-to-track association problem, \revisNew{although it achieves a linear computing complexity in the sum of the number of detections from all the sensors \cite{Vo19msGLMB}}. 
By contrast, the latter {communicates and} fuses the \revisNew{posterior estimates (given in means of various densities) produced by local filters}, leading to the suboptimal but robust fusion result. This often realized in the distributed manner offers better computational scalability and system reliability, while the distributed fusion structure is more flexible and less prone to node failures as compared with the centralized fusion. The so-called arithmetic/geometric average (AA/GA) fusion approach \cite{Li22chapter} 
is up to the task.


 \revisNew{Based on the fundamental linear opinion pool \cite{Stone61LOP}, the AA fusion has} an intimate connection with the mixture-type estimator \cite{Li23BM} 
such as the celebrated Gaussian mixture (GM) filter that is arguably the AA fusion of multiple Gaussian estimators. 
It is also in tune with the average consensus approach \cite{Olfati-Saber07,Sayed14book,Meyer18ProIEEE} suited for distributed implementation.
However, its application for multi-sensor RFS density fusion, both centralized and distributed, has only appeared a few years. Prior to this, the dominant solutions are based on the Bayes or minimum-variance framework. As a non-Bayes method, the AA \revisNew{density} fusion was first known 
due to its conservativeness (namely free of information double-counting \revisNew{regardless of the correlation between sensors}) \cite{Bailey12,Li2021SomeResults}, 
high-efficiency for computing, and superiority to deal with misdetection \cite{Yu16,Li17merging,Li17PC,Li17PCsmc,Gostar17,Li19Bernoulli,Li20AAmb, Gao20cphd,Gostar20,Gao20GLMB,Ramachandran21,Da_Li20TSIPN,Li19Diffusion,Yi20AAfov,GLi20AA-asynchronousPHD,Yu22asynchronousCPHD,GLi22aaPMB,Li23AApmbm}. 
Recently, more properties of the AA/GA fusion have been analyzed in \cite{Koliander22} and fast convergence in \cite{Kayaalp22}. 

Nevertheless, the multi-target RFS density 
that involves the number, states and even identities of the targets is not a simple extension of the classic single-target density. 
The RFS filters \revisNew{aim to} solve the target detection, state-estimation and even data-association jointly, \revisNew{for which the core need of effective fusion} is the accuracy and robustness of the detection and localization of each target.
Despite a good many of RFS-AA fusion formulas and flexible implementations \cite{Yu16,Li17merging,Li17PC,Li17PCsmc,Gostar17,Li19Bernoulli,Li20AAmb, Gao20cphd,Gostar20,Gao20GLMB,Ramachandran21,Da_Li20TSIPN,Li19Diffusion,Yi20AAfov,GLi20AA-asynchronousPHD,Yu22asynchronousCPHD,GLi22aaPMB,Li23AApmbm} 
proposed to accommodate different RFS densities, 
a theoretical-unified analysis 
of these formulas 
is still missing. Existing derivations and reasoning for multi-sensor RFS density fusion \cite{Gostar17,DaKai_Li_DCAI19,Gostar20,Gao20cphd,Gao20GLMB} are mainly based on the best fit of the multi-target probability distribution (MPD). They have \revisNew{three} theoretical disadvantages \revisNew{and limitations}
\begin{itemize}
  \item These derivations rely on the MPD-best-fit property (see section \ref{sec:ordinary-AA-fusion}) which does not directly relate to the real target distribution or the optimal multi-sensor posterior 
and therefore does not disclose how the linear fusion can deal with false/missing data and gain better accuracy in target detection and localization. 
Furthermore, it cannot explain why the GA fusion that has the similar MPD-best-fit property performs very differently in comparison with the AA fusion. 
  \item The MPD-best-fit needs approximation in general because the AA of most \revisNew{forms} of RFS MPDs is no more the same \revisNew{form. Each derivation is therefore only applicable to a specific form of MPD. 
  They do} not form an exact, unified framework suited for all RFS filters.
  \item \revisNew{A complete definition of the divergence between labeled densities that needs to take into account the divergence of both the state distribution and the labels is still missing. 
Existing definitions, as well as the derivation based on them, make little account of the difference between matched labels and apply only when the number of labels is the same among the fusing densities.}
\end{itemize}



To overcome the above limitations and disadvantages, this paper contributes in the following aspects:
\begin{enumerate}
\item We propose a statistical concept, \textit{PHD consistency}, which lays a technically unified, intuitively meaningful foundation for (labeled) RFS fusion to explain how the PHD-AA fusion can deal with false/missing data and improve accuracy in target detection and localization. 
\item We derive appropriate AA fusion formulas for different (labeled) RFS filters starting from the (labeled) PHD-AA fusion, which all maintain exactly AA calculation in terms of the relevant (labeled) PHDs and ensure PHD consistency. Although the final fusion results turn out to coincide with those separately given in \cite{Gostar17,DaKai_Li_DCAI19,Gostar20,Gao20cphd,Gao20GLMB}, 
    our derivation is based on an exact, unified framework that is suited for all RFS filters. This unified framework paves a flexible way for \revisNew{either homogeneous or heterogenous RFS filter fusion}; 
    see the \arxivNew{companion papers \cite{Li23Heterogeneous,Li23HeterVA}}. 
    It also exposes the essence of existing consensus-driven multi-sensor RFS filters which merely seek consensus over their (labeled) PHD rather than over the MPDs, \revisNew{while the MPD-best-fit derivation implies the consensus on the level of MPD over-optimistically}. 
\item
    \revisNew{Aiming at labeled PHD consistency, the averaging of labeled PHDs that are functions on single target, single label 
    is appropriate for the general case in which the fusing filters have different numbers of labels, overcoming the deficiency of the labeled MPD-best-fit derivation.} 
\item Derived from the labeled PHD-AA fusion, the exact AA fusion of the \revisNew{general} ($\delta$-)GLMB \cite{Vo13Label,Vo14GLMB} is given for the first time, 
    which can be easily tailored in more specific forms to accommodate the marginalized GLMB (M-GLMB) \cite{Fantacci15M-glmb} and the labeled MB (LMB) \cite{Reuter14LMB}.
\end{enumerate}

\revisNew{This paper is the second part of a series of papers that aim to provide a comprehensive and  thorough study of the AA fusion methodology and its application for target tracking. The first is on the fundamental Kalman/particle filter fusion \cite{Li2021SomeResults} \arxivNew{while the third and the forth} are on advanced heterogeneous unlabeled and labeled RFS filter fusion \cite{Li23Heterogeneous,Li23HeterVA} based on the finding of this work. The rest of this paper} is organized as follows. The representative RFS processes and the AA fusion approach are briefly introduced in section \ref{sec:background}. The \revisNew{concept of} PHD consistency is formally defined in section \ref{PHD-Consistency} where the suitable AA fusion formulas for PHD/CPHD filters are consequently derived. The unlabeled Bernoulli, MB and MB mixture (MBM), and the labeled GLMB, M-GLMB and LMB fusion are addressed in sections \ref{sec:BMs} and \ref{sec:label}, respectively. 
The paper is concluded in section \ref{sec:conclusion}.

%

\section{Background and Preliminary} \label{sec:background}
\revisNew{The essential goal of multi-sensor cooperative target tracking is to get more accurate and robust detection and localization of the present targets based on the information exchange and fusion among the sensors. In this work, we focus on the estimate-level fusion where what is fused is the first order moments of the MPDs, namely the PHD, individually calculated by the local RFS filters.}
{For brevity,}
the details of local single-sensor filters and {of} the inter-sensor communication are fully omitted. We further omit the time notation and the dependence of the filter estimate on the observation process. Sensors are assumed to be synchronous and perfectly coordinated. 

\subsection{RFS Theory}
Denote by $\mathcal{X} \subseteq \mathbb{R}^d$ the $d$-dimensional state space and by $\mathbb{X}$ all of the finite subsets of $\mathcal{X}$.
The states of a random number of targets are described by an RFS ${X} = \big\lbrace \mathbf{x}_{1}, \dots, \mathbf{x}_{n} \big\rbrace \subseteq \mathbb{X}$, where $n =|{X}|$ denotes the random number of targets, namely the cardinality of the set,
and $\mathbf{x}_{i}\in \mathcal{X}$ is the state vector of the $i$-th target.
The random nature of the multi-target RFS ${X}$ is captured by its MPD, denoted by $f({X})$. For any realization of $X$ with a given cardinality $|{{X}}| =n$, \revisNew{denoted by} ${X}_n = \big\lbrace \mathbf{x}_{1}, \dots, \mathbf{x}_{n} \big\rbrace$, \cite[Eq.2.36]{Mahler14book}
\begin{equation}\label{eq:fisst_mttPDF}
  f({X}_n) = n! \rho(n) f(\mathbf{x}_{1}, \dots, \mathbf{x}_{n} )
\end{equation}
where the localization densities $f(\mathbf{x}_{1}, \dots, \mathbf{x}_{n} )$ for $n=1,2,...,$ are symmetric in their arguments and the cardinality distribution $\rho(n)\triangleq \mathrm{Pr}\{|{{X}}|=n\} = \int_{|{X}| = n} {f({X})\delta {X}}$ is given by 
\begin{align}
  \rho(n) 
  & = {\frac{1}{{n!}}\int_{\mathcal{X}^n} {f\big(\{ {\mathbf{x}_1},\dots,{\mathbf{x}_n}\} \big)d{\mathbf{x}_1}\dots d{\mathbf{x}_n}}}
\end{align}

The set integral in $\mathbb{X}$ is defined as \cite[Ch. 3.3]{Mahler14book}
\begin{align}
\int_{\mathbb{X}} {f({X})\delta {X}}
   &= 
    \sum\limits_{n = 0}^\infty  {\frac{1}{{n!}}\int_{\mathcal{X}^n} {f\big(\{ {\mathbf{x}_1},\dots,{\mathbf{x}_n}\} \big)d{\mathbf{x}_1}\dots d{\mathbf{x}_n}}} \label{eq:set-integral-expansion} \\
   &= \sum\limits_{n = 0}^\infty \rho(n) \\
   &\revisNew{ =1 \nonumber}
\end{align}
where $f(\emptyset )= \rho(0)$.

The PHD $D(\mathbf{x})$, also known as the first moment density \cite[pp. 168-169]{Goodman97},
of the MPD $f({X})$ is a density function on single target $\mathbf{x}\in X$, defined as \cite[Ch.4.2.8]{Mahler14book} 
\begin{align}
D(\mathbf{x})   
& \triangleq \int_\mathbb{X}   {\bigg(\sum_{\mathbf{y}\in X}{\delta_\mathbf{y}}(\mathbf{x}) \bigg)f(X)\delta X} \label{def-PHD}
\end{align}
where \revisNew{${\delta_\mathbf{y}}(\mathbf{x})$ is the Dirac delta function concentrated at $\mathbf{y}$.} 

The PHD has clear physical significance as its integral in any region $\mathcal{S} \subseteq \mathcal{X} $ gives the expected number $\hat{N}^{\mathcal{S}}$ of targets in that region, i.e.,
\begin{equation}\label{eq:def-phd-integral=hatN}
  \hat{N}^{\mathcal{S}} = \int_{\mathcal{S}} D(\mathbf{x}) d \mathbf{x}
\end{equation}
Therefore, it tells how well the present targets are detected.

\subsection{Some Important \revisNew{RFSs}} \label{sec:Classic-RFS-distributions}

\subsubsection{Independent and Identically Distributed Cluster (IIDC)}

An RFS $ X $ is considered as an IIDC RFS if its cardinality is arbitrarily distributed according to $ \rho(n) $, and for any finite cardinality $n$, the elements are independently and identically distributed according to the single-target probability density (SPD) $s(\mathbf{x})$. 
 The MPD and PHD associated with the IIDC RFS $X$ are, respectively, given by
\begin{align}
	f^{\text {iidc}}(X_n) &= n!\rho(n)\prod\limits_{\mathbf{x} \in X_n} {s(\mathbf{x})} \label{eq:IIDC-p} \\  %
	D^{\text {iidc}}(\mathbf{x}) &= \sum\limits_{n = 0}^\infty  {n\rho(n)} {s}(\mathbf{x}) \label{eq:iidc-phd}  
\end{align}

The IIDC RFS reduces to a Poisson RFS when the cardinality distribution $\rho(n)$ is Poisson. The MPD and PHD of the Poisson RFS $X_n$ with mean $\lambda$ are, respectively, given by
\begin{align}
	f^{\text {p}}(X) =\,\,& {{\text{e}}^{ - {\lambda}}}\prod\limits_{\mathbf{x} \in X} {{\lambda}{s}(\mathbf{x})}  \label{eq:Poisson-PD} \\
	D^{\text {p}}(\mathbf{x}) =\,\,& \lambda {s}(\mathbf{x}) \label{eq:Poisson-PHD}
\end{align}

\subsubsection{Bernoulli and MB}
The density and PHD of a Bernoulli RFS $X$ 
with target existence probability $r$ and SPD $s(\mathbf{x})$ are, respectively, given by
\begin{align}
f^{\text{b}}\left(X\right) & = \begin{cases}
1-r & X  =\emptyset\\
rs\left(\mathbf{x}\right) & X =\left\{ \mathbf{x}\right\} \\
0 & \mathrm{otherwise}
\end{cases} \label{eq:Bernoulli-PD} \\
D^{\text{b}}\left(\mathbf{x}\right) & = rs\left(\mathbf{x}\right) \label{eq:Bernoulli-PHD}
\end{align}

An MB RFS $ X_n$ is the union of $n$ independent Bernoulli RFSs  \cite{Vo09CBmember}. 
Denoting the $l$-th Bernoulli component (BC) by $\big(r_{l},s_{l}(\mathbf{x})\big)$, the MPD and PHD of MB RFS $ X_n$ are, respectively, given by
\begin{align}
	f^\text{mb} (X_n) & = \sum_{\uplus_{l=1}^{n}X^{(l)}=X_n}\prod_{l=1}^{n}f^\text{b} \left(X^{(l)}\right) \label{eq:def_MB}\\
	D^{\text {mb}}(\mathbf{x}) & = \sum_{l=1}^n{r_l}{s_l}(\mathbf{x})  \label{eq:PHD-MB}
\end{align}
where $\uplus$ denotes the disjoint union. 

When multiple MBs are used jointly and associated with different measurement-to-track hypotheses or labels, they result in the PMBM \cite{Williams15taesPMBM}, the MBM \cite{Angel18PMBMdeivation} and the GLMB \cite{Vo13Label,Vo14GLMB} which are all conjugate prior for multi-target filtering. 

\subsubsection{PMBM}
An PMBM RFS\cite{Williams15taesPMBM} is the convolution of a Poisson RFS and an MBM RFS with MPD
\begin{equation}
	f^{\text {pmbm}}(X) =\sum\limits_{Y \uplus W =X} {{f^{\text {p}}}(Y){f^{\text {mbm}}}(W)}
\end{equation}
where the Poisson density $f^\text{p}(\cdot)$ represents all undetected targets, and the MBM density $f^{\text{mbm}}(\cdot)$ representing potentially detected targets is a weighted sum of MBs and has the form
\begin{equation}
	{f^{\text {mbm}}}(X_n) \propto \sum\limits_{j \in \mathbb{J}} {\sum\limits_{\uplus_{l=1}^{n}X^{(l)} =X_n} {\prod\limits_{l= 1}^n {{w_{j,l}}f_{j,l}^{\text {b}}\left({X^{(l)}}\right)} } } \label{eq:def_MBM_MPD}
\end{equation}
where $j$ is an index to the MBs (each corresponding to a global hypothesis), $\mathbb{J}$ is the index set and $w_{j,l}$ is a coefficient/weight assigned to the $l$-th BC in the $j$-th hypothesis/MB.

The PHD of the MBM RFS $X_n$ of the MPD \eqref{eq:def_MBM_MPD} is
\begin{equation}
	D^{\text {mbm}}(\mathbf{x}) = \sum\limits_{j \in \mathbb{J}} \sum_{l=1}^n w_{j,l}r_{j,l}s_{j,l}(\mathbf{x}) \label{eq:def_MBM_PHD}
\end{equation}
where $\big(r_{j,l},s_{j,l}(\mathbf{x})\big)$ denote the parameters of the $l$-th BC in the $j$-th hypothesis/MB.


As a special PMBM, the Poisson MB (PMB) \cite{Williams15TSPpmb} uses the MB 
instead of MBM for representing detected targets. 

\subsubsection{GLMB}
Denote by $\mathcal{L}$ the label space and by $\mathbb{L}$ all of the finite subsets of $\mathcal{L}$.
A labeled RFS (LRFS) is an RFS whose elements are assigned distinct, discrete labels \cite{Vo13Label,Vo14GLMB}. 
A realization of an LRFS with cardinality $n$, multi-target state $X_n$ and label set $L_n = \big\lbrace l_{1}, l_{2}, \dots, l_{n} \big\rbrace \subseteq \mathbb{L}$ is denoted by $\widetilde{X}_n = \{ (\mathbf{x}_1, l_1),(\mathbf{x}_2, l_2),...,(\mathbf{x}_n, l_n) \} \subseteq \mathbb{X} \times \mathbb{L}$. The LRFS is completely characterized by its MPD $\pi\big(\widetilde{X}\big)$. 
The most known LRFS density, namely the GLMB, is a mixture of multi-target exponentials with MPD \cite{Vo13Label,Vo14GLMB}
\begin{equation}\label{eq:glmb-in-general}
	\pi^{\text {gl}}\big(\widetilde{X}\big) =  \Delta \big(\widetilde{X}\big) \sum\limits_{c \in \mathbb{C}} {{\omega ^{(c)}}\Big(\mathcal{L}\big(\widetilde{X}\big)\Big)\prod\limits_{(\mathbf{x},l) \in \widetilde{X}} {{s^{(c)}}(\mathbf{x},l)} }
\end{equation}
where the distinct label indicator $\Delta \big(\widetilde{X}\big)= \delta_{|\widetilde{X}|}\big[\mathcal{L}\big(\widetilde{X}\big)\big]$, $ \mathbb{C} $ is a discrete index set, the projection function $\mathcal{L}(\mathbb{X} \times \mathbb{L}) \rightarrow \mathbb{L}$ is given by $\mathcal{L}\big((\mathbf{x},l)\big) = l$, $\omega ^{(c)}(L_n) $ is the weight of the hypothesis that $|\widetilde{X}|=n$ targets are presented with respective labels $ l_1, \ldots, l_n \in L_n $ and respective labeled SPDs $ s^{(c)}(\mathbf{x}_1,l_1),\ldots, s^{(c)}(\mathbf{x}_n,l_n)$, satisfying
$$\sum\limits_{L \subseteq \mathbb{L}} {\sum\limits_{c \in \mathbb{C}} {{\omega^{(c)}}(L)} }  = 1, \int_{\mathcal{X}} {{s^{(c)}}(\mathbf{x},{l})d\mathbf{x}}  = 1.$$
%
%

The labeled PHD (LPHD) of an GLMB is given by \cite{Vo13Label}
\begin{align}
  {\widetilde D}^{\text {gl}}(\mathbf{x},l) & \triangleq \int_{\mathbb{X} \times \mathbb{L}}  {\pi}\big((\mathbf{x},l)\cup {\widetilde X}\big)\delta {\widetilde X} \\ 
  &= \sum\limits_{c \in \mathbb{C}}{{s^{(c)}}(\mathbf{x},l)}\sum\limits_{L \subseteq \mathbb{L}} {{\mathrm{1}_L}(l) {\omega ^{(c)}}(L)}  \label{eq:glamb-Lphd}
\end{align}
where ${\mathrm{1}_x}(y) = 1$ if $y \subseteq x$ and ${\mathrm{1}_x}(y) = 0$ otherwise.

%

A $\delta$-GLMB RFS is a special GLMB RFS with $\mathbb{C} =  \mathbb{L} \times \Xi$, ${\omega^{(c)}}(L') = {\omega^{(L,\xi )}}{\delta _L}(L')$, and ${s^{(c)}}(\mathbf{x},l) = {s^{(L,\xi )}}(\mathbf{x},l) = {s^{(\xi )}}(\mathbf{x},l)$, where the pair $(L,\xi)$ represents the hypothesis that the track set $L$ has a history $\xi$ of measurement-track association maps \cite{Vo14GLMB}. The MPD of a $\delta$-GLMB RFS {$\widetilde{X}$} is
\begin{equation} 
	\pi^{{\delta}}\big(\widetilde{X}\big) 
= \Delta \big(\widetilde{X}\big) \sum\limits_{L \subseteq \mathbb{L}} {\delta _L}\big[\mathcal{L}\big(\widetilde{X}\big)\big]  \sum\limits_{\xi \in \Xi} {{\omega^{(L,\xi )}}\prod\limits_{(\mathbf{x},l) \in \widetilde{X}} {{s^{(\xi)}}(\mathbf{x},l)}} \nonumber
\end{equation}

The 
LPHD 
of the $\delta$-GLMB is given as follows \cite{Vo13Label}
\begin{align}
  {\widetilde D}(\mathbf{x},l) & \triangleq \int  {\pi}\big((\mathbf{x},l)\cup {\widetilde X}\big)\delta {\widetilde X} \label{eq:def-LRFS-phd-label} \\
  &=  \sum\limits_{L \subseteq \mathbb{L}} {{\mathrm{1}_L}(l) \sum\limits_{\xi \in \Xi} {\omega^{(L,\xi )}} {{s^{(\xi)}}(\mathbf{x},l)} } \label{eq:glamb-phd} 
\end{align}

\subsection{Ordinary AA Density Fusion} \label{sec:ordinary-AA-fusion}
We briefly review the AA fusion in terms of the ordinary probability density function (PDF). 
For a set of PDFs $f_i(\mathbf{x})$ produced by estimator $i \in \mathcal{I}$ where $\mathcal{I} =\{1,2,...,I\}$ denotes the set of sensors, the AA fusion is simply given as follows
\begin{equation}\label{eq:AA-density}
{f_{\text{AA}}}(\mathbf{x}) \triangleq \sum\limits_{i \in {\mathcal{I}}} {{w_i}{f_i}(\mathbf{x})}
\end{equation}
where the fusion weights $\mathbf{w} =\{w_1,\dots,w_I\} \in \mathbb{W}$, and weight space $\mathbb{W} \triangleq\{\mathbf{w} \in \mathbb{R}^{I}|\mathbf{w}^\mathrm{T}\mathbf{1}_I = 1, w_i > 0, \forall i \in \mathcal{I}\} \subset \mathbb{R}^{I}$.

The AA fusion greatly 
ensures the desired fusion closure for these mixture densities.
Another important property of the AA PDF fusion is that
{the expectation $\mathbb{E}_{p_\text{AA}(\mathbf x | \mathbf w)}{H(\mathbf x)}$ of an arbitrary function $H(\mathbf x)$ of $\mathbf x$ with respect to the AA fused density $f_\text{AA}(\mathbf x | \mathbf w) $ is given as follows
\begin{equation}\label{eq:exp_AA_H_func}
        \mathbb{E}_{f_\text{AA}(\mathbf x | \mathbf w)}{H(\mathbf x)}
        = \sum_{i\in \mathcal{I}} w_{k} \mathbb{E}_{f_i (\mathbf x )}{H(\mathbf x)}
\end{equation}}

{Arbitrary moments of $\mathbf x$ can be obtained for appropriate choices of $H(\mathbf x)$. When $ H(\mathbf x)= \mathbf{x}$, \eqref{eq:exp_AA_H_func} reads the unbiasedness of the AA fusion, i.e., the AA of a number of unbiased estimators is guaranteed to be unbiased.}
It was further pointed out that the AA is a Fr\'{e}chet mean in the sense of $L_2$ norm error \cite{Li20AAmb}, i.e.,
\begin{equation}
  f_\mathrm{AA}(\mathbf{x})  = \operatorname*{arg\,min}_{g\in\mathcal{F}_{\mathcal{X}}} \sum\limits_{i \in {\mathcal{I}}}{w_i\int_{\mathcal{X}} \big(f_i(\mathbf{x})-g(\mathbf{x})\big)^2 d\mathbf{x}}   \label{eq:FrechetAA}
\end{equation}
where $\mathcal{F}_{\mathcal{X}}$ denotes the set of scalar-valued functions in space $\mathcal{X}$: $\mathcal{F}_{\mathcal{X}} = \{f: \mathcal{X} \rightarrow \mathbb{R} \}$.

The AA fusion minimizes the weighted sum of the Kullback-Leibler (KL) divergences of the fusing densities relative to the fused result \cite{Kulhavy96,Li2021SomeResults}, 
i.e., 
\begin{equation}\label{eq:AA-density-KLD}
 {f_{\text{AA}}}(\mathbf{x}) = \operatorname*{arg\,min}_{g\in\mathcal{F}_{\chi}} \sum\limits_{i \in {\mathcal{I}}} {w_iD_\text{KL}\big(f_i\|g\big)}
\end{equation}
where the KL divergence ${D_{{\text{KL}}}}(f\|g) = \int_{\mathcal{X}} {f(\mathbf{x})\log \frac{f(\mathbf{x})}{g(\mathbf{x})} d \mathbf{x}}$ in the case of continuous functions and $D_{\text{KL}}(f\|g) = \sum_{n=0}^\infty f(n)\log \frac{f(n)}{g(n)}$ in the case of discrete functions. 

The above optimization, also known as 
minimum information loss (MIL) fusion \cite{Gao20cphd,Gao20GLMB}, provides a best fit of the mixture (BFoM) of all fusing distributions with weights, $\{w_if_i\}_{i\in {\mathcal{I}}}$. Considering that the mixture $\{f_i\}_{i\in {\mathcal{I}}}$ contains the information obtained by all sensors, it is a reasonable substitute of ``\textit{all information}''. This approach 
best fits the mixture but {not the true distribution 
or the multi-sensor optimal distribution} and so the BFoM/MIL itself does not guarantee any gain in the estimation accuracy or robustness. This is easily evidential. By turning the direction of the KL divergence in \eqref{eq:AA-density-KLD}, the minimization will lead to the GA fusion \cite{Kulhavy96,DaKai_Li_DCAI19} which is significantly different with the AA fusion, 
i.e.,
\begin{align}
 {f_{\text{GA}}}(\mathbf{x}) & = \operatorname*{arg\,min}_{g: \int g(\mathbf{x})d\mathbf{x} =1} \sum\limits_{i \in \mathcal{I}} {w_iD_\text{KL}\big(f_i\|g\big)} \label{eq:GA-density-KLD} \\
 & =\frac{1}{c} \prod_{i \in \mathcal{I}} f_i^{w_i}(\mathbf{x})
\end{align}
where $c= \int_{\mathcal{X}} \prod_{i \in \mathcal{I}} f_i(\mathbf{x})^{w_i} d\mathbf{x}$ is a normalization constant.
%


However, given the true distribution $p(\mathbf{x})$, the optimal fusing weights should be given by
\begin{align}
  \mathbf{w}_\text{opt} = & \mathop{\arg\min}\limits_{\mathbf{w} \in \mathbb{W}}  D_\text{KL}\left({f_\text{AA}}\| p\right) \nonumber \\
  = &  \mathop{\arg\min}\limits_{\mathbf{w} \in \mathbb{W}} \sum_{i \in \mathcal{I}} w_i \big( D_\text{KL}({f_i}\| p)  - D_\text{KL}( {f_i}\| f_\text{AA}) \big) \label{eq:min_W-BFoM}
\end{align}
This formulation is solidly motivated since the true distribution $p(\mathbf{x})$ is referred to, which is a fundamentally different to \eqref{eq:AA-density-KLD} and \eqref{eq:GA-density-KLD}.
But as the true distribution $p(\mathbf{x})$ is actually unknown, a practically operable alternative is given by ignoring the former part in \eqref{eq:min_W-BFoM} and solving the following suboptimal, \textit{diversity preference}, maximization problem \cite{Li2021SomeResults}
\begin{equation}
  \mathbf{w}_\text{BFoM} =  \mathop{\arg\max}\limits_{\mathbf{w} \in \mathbb{W}}  \sum_{i \in \mathcal{I}} w_i D_\text{KL}( {f_i}\| f_\text{AA}) \label{eq:suboptimal-AAweight}
\end{equation}

Combining \eqref{eq:AA-density-KLD} with \eqref{eq:suboptimal-AAweight} yields a maxmin formulation for joint optimization of the fusing form and fusing weights
\begin{equation}
(\mathbf{w}_\text{BFoM},f_\text{AA}) = 
  \arg\mathop{\max}\limits_{\mathbf{w} \in \mathbb{W}} \mathop{\min}\limits_{g \in \mathcal{F}_{\mathcal{X}}} \sum_{i \in \mathcal{I}} w_i D_\text{KL}( {f_i}\| g) \label{eq:JiontOpt-AA}
\end{equation}
which is hereafter referred to as the BFoM optimization.

\revisNew{We note that the above BFoM formulation is appropriate for any ordinary density including the PHD functions but not the labeled densities in general; see Section \ref{sec:LRFS-label}.}

\section{PHD-Consistency and PHD-AA Fusion} \label{PHD-Consistency}
\revisNew{An intuitive idea to resolve RFS density fusion is to extend the PDF-AA fusion to the MPD domain for which, however, the major barrier in theory is from the fact that most RFS MPDs do not admit closed-form AA fusion, e.g., the AA of IIDC/sMBs is no more an IIDC/MB \cite{Li20AAmb,Da_Li20TSIPN}. This mainstream idea in existing derivations \cite{Gostar17,DaKai_Li_DCAI19,Gostar20,Gao20cphd,Gao20GLMB} is actually undesired to the local filters that need to maintain a closed-form filter-iteration. To fulfil the need of preserving the \revisNew{form} of the fusing MPD, ad-hoc approximation is often needed for fitting the MPDs \cite{DaKai_Li_DCAI19,Gostar20,Gao20cphd,Gao20GLMB}. In spite of this, only homogeneous MPDs are able to be fused. To overcome this problem, we resort to fusion in terms of the first moment density of the MPD rather than the MPD itself. In this section, we first derive the PHD-AA fusion from the MPD-AA fusion in this section. Then, a statistical concept, \textit{PHD consistency}, is proposed which explains how the PHD-AA fusion can combat false/missing data and improve the accuracy in target detection and localization. Finally, the proper, exact formulation of the AA fusion of the PHD filter and of the CPHD filter is derived.}

\subsection{From MPD-AA to Cardinality-AA and PHD-AA}
The MPD-AA fusion is given by 
\begin{equation}\label{eq:RFS-AA-Whole}
{f_{\text{AA}}}({X}) \triangleq \sum\limits_{i \in {\mathcal{I}}} {{w_i}{f_i}({X})}
\end{equation}
which corresponds to the following BFoM, c.f., \eqref{eq:AA-density-KLD},
\begin{equation}\label{eq:RFS-AA-Whole-KLD}
 {f_{\text{AA}}}({X}) = \operatorname*{arg\,min}_{g\in\mathcal{F}_\mathbb{X}} \sum\limits_{i \in {\mathcal{I}}} {w_iD_\text{KL}\big(f_i\|g\big)}
\end{equation}
where $\mathcal{F}_\mathbb{X} \triangleq \{f: \mathbb{X} \rightarrow \mathbb{R} \}$ and the set KL divergence is defined as
\begin{equation}\label{eq:def_set-kld}
  D_\text{KL}\big(f\|g \big) \triangleq \int_\mathbb{X} {f(X)\log \frac{f(X)}{g(X)} \delta X}
\end{equation}


Notably, according to the definition of the PHD \eqref{def-PHD}, the MPD-AA fusion \eqref{eq:RFS-AA-Whole} implies cardinality-AA fusion and PHD-AA fusion, respectively, i.e.,
\begin{align}
 \rho_{\text {MPD-AA}}(n) & = \int_{|{X}| = n} {f_{\text{AA}}({X})\delta {X}} \nonumber \\
 & = \sum\limits_{i \in {\mathcal{I}}} {w_i} \int_{|{X}| = n} {f_{i}({X})\delta {X}} \nonumber \\
 & =   \sum\limits_{i \in {\mathcal{I}}} {w_i} \rho_{i}(n) \nonumber \\
 & \triangleq \rho_{\text {AA}}(n) \label{eq:def-rho-AA} \\
D_{\text {MPD-AA}}(\mathbf{x}) & = \int_\mathbb{X}   {\Big(\sum_{\mathbf{w}\in X}{\delta_\mathbf{w}}(\mathbf{x}) \Big)f_{\text{AA}}(X)\delta X} \nonumber \\
& =  \sum\limits_{i \in {\mathcal{I}}} {{w_i}{D_i}(\mathbf{x})} \nonumber \\
& \triangleq D_{\text {AA}}(\mathbf{x}) \label{eq:def-PHD-AA}
\end{align}
which corresponds to the following BFoM, c.f., \eqref{eq:RFS-AA-Whole-KLD},
\begin{align}
\rho_{\text {AA}}(n) & = \operatorname*{arg\,min}_{\rho:  \mathbb{N}\rightarrow \mathbb{R}} \sum\limits_{i \in {\mathcal{I}}} {w_i D_\text{KL}\big(\rho_i\|\rho\big)} \label{eq:rho-AA-KLD} \\
 D_{\text{AA}}({\mathbf{x}}) & = \operatorname*{arg\,min}_{g\in \mathcal{F}_{\mathcal{X}}} \sum\limits_{i \in {\mathcal{I}}} {w_i D_\text{KL}\big(D_i\|g\big)} \label{eq:PHD-AA-KLD}
\end{align}

The cardinality and PHD are of clear physical significance and their average calculation has obvious benefit as explained next. Even more importantly, the PHD is not limited to any specific \revisNew{RFS density forms as the MPD fusion does and so the PHD-AA fusion can be straightforwardly applied to any RFS filters.} 
Obviously, the MPD average/consensus implies the PHD average/consensus but not vice versa.

\subsection{Motivation for PHD-AA Fusion} \label{sec:Motivation-phd-AA-fusion}
The PHD is a key statistical character of the RFS distribution, and so 
an accurate PHD estimate is vital in any RFS filters. On this point, we present the following definition relative to the effectiveness of the multi-sensor fusion. 
\begin{definition}[PHD Consistency] The accuracy of the multi-sensor fused PHD increases with the number of fusing estimators. 
That is, the PHD consistent fusion should result in statistically more accurate detection of the targets and the fused PHD converges to the first moment of the true MPD as the number of the fusing estimators increases indefinitely. 
\end{definition}

We note that our above definition abides with the standard definition of the consistency \revisNew{for parameter estimation}, i.e., as the number of data points used increases indefinitely, the resulting sequence of estimates converges in probability to the ground truth. 
This is different from the concept given in \cite{Julier06consistency,Uney19consistency}, which means consenus/agreement. 

\begin{lemma} \label{lemma:phd-aa-consistency}
\revisNew{Given that the estimates of the number of targets yielded} by the fusing estimators are conditionally independent with each other and unbiased everywhere in the state space,
the PHD-AA fusion \eqref{eq:def-PHD-AA} \revisNew{guarantees} 
PHD consistency.
\end{lemma}

\begin{proof}
Denoting by $\text{MSE}_{i}^{\mathcal{S}}$ the mean-square-error (MSE) of the estimated number $\hat{N}_{i}^{\mathcal{S}}$ of targets in any region $\mathcal{S}\subseteq \mathcal{X}$ obtained by filter $i \in {\mathcal{I}}$, the estimated number of targets given by the integral of the AA fused PHD in $\mathcal{S}$ is, \revisNew{c.f. \eqref{eq:exp_AA_H_func}}
$$\hat{N}_\text{PHD-AA}^{\mathcal{S}} = \int_{\mathcal{S}} D_\text{AA}(\mathbf{x})d \mathbf{x} = \int_{\mathcal{S}} \sum\limits_{i \in {\mathcal{I}}} {w_i D_i(\mathbf{x})}  d \mathbf{x} =  \sum\limits_{i \in {\mathcal{I}}} {w_i \hat{N}_i^{\mathcal{S}}}$$
The PHD consistency is ensured in the following sense
\begin{align}
  \text{MSE}_\text{PHD-AA}^{\mathcal{S}} & = \mathbb{E}\big[(\hat{N}_\text{PHD-AA}^{\mathcal{S}}-{N}^{\mathcal{S}})^2\big] \nonumber \\
  & = \mathbb{E}\Big[\big(\sum\limits_{i \in {\mathcal{I}}} {w_i \hat{N}_i^{\mathcal{S}}}-{N}^{\mathcal{S}}\big)^2\Big] \nonumber \\
  & = \sum\limits_{i \in {\mathcal{I}}}w_i^2\mathbb{E}\Big[\big( { \hat{N}_i^{\mathcal{S}}}-{N}^{\mathcal{S}}\big)^2\Big] \label{eq:indepent-card-est} \\
  & = \sum\limits_{i \in {\mathcal{I}}}w_i^2 \text{MSE}_i^{\mathcal{S}} \label{eq:phd-AA-accuracy}
\end{align}
where expectation $ \mathbb{E}[\cdot]$ is conditional on the observation, $N^{\mathcal{S}}$ denotes the real number of targets in region $\mathcal{S}$, and conditional independence and unbiasedness were used in \revisNew{\eqref{eq:indepent-card-est}}.

By using $w_i= \frac{1}{I}, \forall i  \in {\mathcal{I}}$, \eqref{eq:phd-AA-accuracy} reduces to
\begin{align}\label{eq:mse-I-mean}
  \text{MSE}_\text{AA}^{\mathcal{S}} &= \frac{1}{I} \sum\limits_{i \in {\mathcal{I}}}w_i \text{MSE}_i^{\mathcal{S}} \\
   & \leq  \frac{1}{I}  \max\{\text{MSE}_i^{\mathcal{S}}\}_{i  \in {\mathcal{I}}}
\end{align}
where the equation holds iff $\text{MSE}_i^{\mathcal{S}} = \text{MSE}_j^{\mathcal{S}}, \forall i, j \in {\mathcal{I}}$ and the consistency is indicated by
\begin{equation}
  \lim_{I\rightarrow\infty} \text{Pr}[\text{MSE}_\text{PHD-AA}^{\mathcal{S}} >\varepsilon]=0 \label{eq:MSEconsistency}
\end{equation}
for every $\varepsilon>0$
\end{proof}

The PHD-AA fusion can be illustrated in Fig. \ref{fig:phd-aa}. 
Beyond the PHD/CPHD filters that directly calculate the PHD in time series, many \revisNew{RFS/LRFS filters are obtained based on PHD matching}, e.g., \cite{Reuter14LMB,Fantacci15M-glmb,wang16MB-gci,Lu17MomentLRFS,Saucan17MS_MB,Da_Li20TSIPN}. The PHD accuracy plays the same important role in these filters.
In addition, the cardinality-AA fusion \revisNew{has the similar property}.


\begin{figure}
  \centering
  \includegraphics[width=6.5cm]{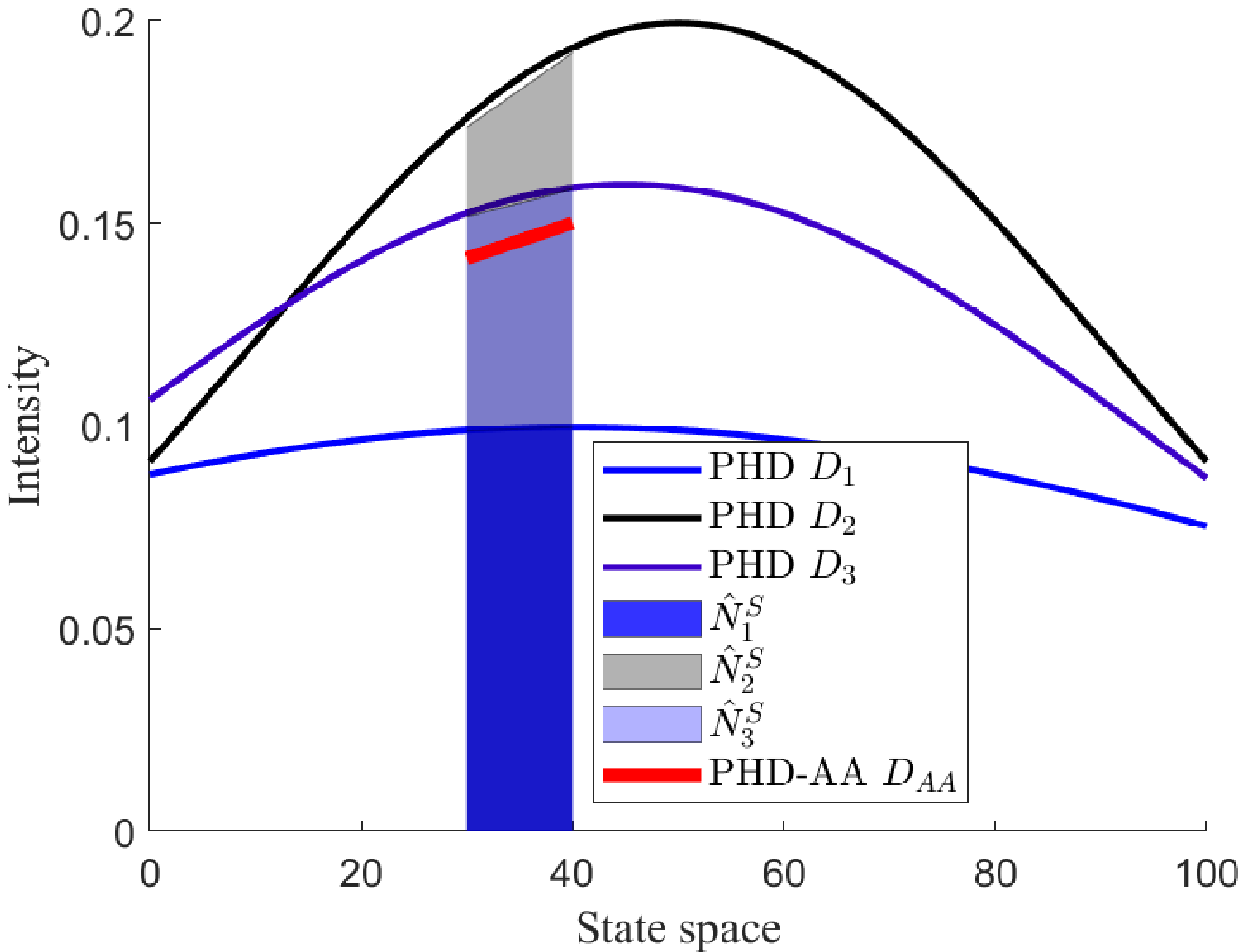}\\
  \caption{The PHD-AA fusion will lead to statistically more accurate and robust detection of the present targets when the estimate of each filter is statistically unbiased and independent on each other. $\hat{N}_i^{\mathcal{S}}$ denotes the estimated number of targets in region $\mathcal{S}$ yielded by RFS filter $i=1,2,3$. }\label{fig:phd-aa}
  \vspace{-3mm}
\end{figure}

\begin{figure}
  \centering
  \includegraphics[width=6.5cm]{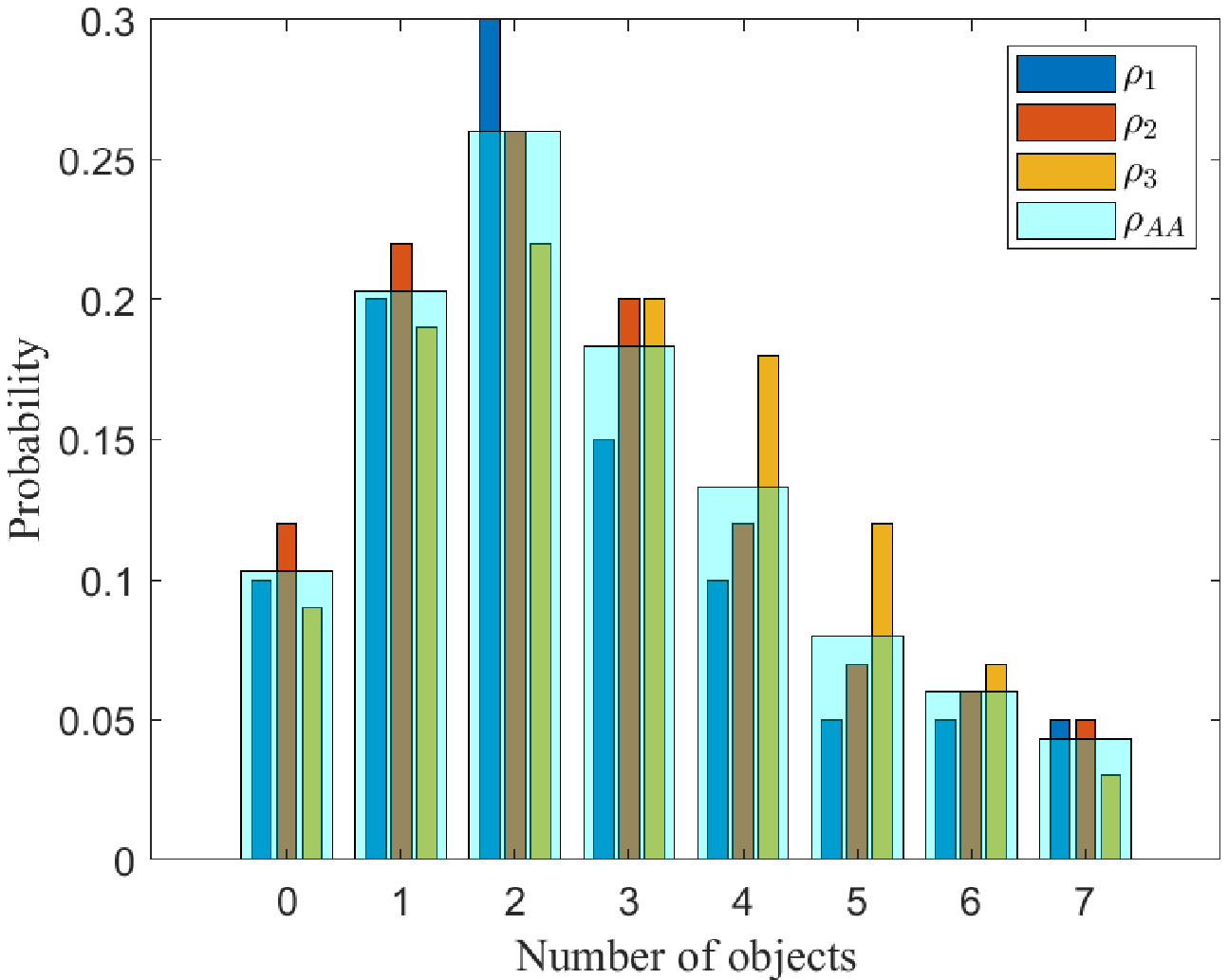}\\
  \caption{The cardinality-AA fusion will lead to statistically more accurate estimate of the number of targets when the cardinality estimate of each filter is statistically unbiased and independent on each other. $\rho_i$ denotes the estimated cardinality \revisNew{probability mass} function yielded by RFS filter $i=1,2,3$. }\label{fig:cardinality-aa}
  \vspace{-3mm}
\end{figure}

\begin{corollary}
Assuming that the fusing cardinality estimates $\rho(n)$ are conditionally independent with each other and unbiased $\forall n=0,1,2,...$, the AA-fused cardinality \eqref{eq:def-rho-AA} with properly designed fusing weights will gain more accurate estimate in the sense of smaller MSE.
\end{corollary}

\begin{proof}
The proof is similar to that for Lemma \ref{lemma:phd-aa-consistency}.
Denote by $\text{MSE}_{i}^{(n)}$ the MSE of the estimated probability $\rho_{i}{(n)}$ of $n$ targets in total obtained by filter $i \in {\mathcal{I}}$. Given that $\rho_{i}{(n)},i \in \mathcal{I}$ are conditionally statistically independent with each other and unbiased, the cardinality-AA fusion leads to
\begin{align}
  \text{MSE}_\text{AA}^{(n)} & = \mathbb{E}\big[(\rho_\text{AA}{(n)}- \text{Pr}[N=n])^2\big] \nonumber \\
  & = \mathbb{E}\Big[\big(\sum\limits_{i \in {\mathcal{I}}} {w_i \rho_i{(n)}}-\text{Pr}[N=n]\big)^2\Big] \nonumber \\
  & = \sum\limits_{i \in {\mathcal{I}}}w_i^2\mathbb{E}\Big[\big( { \rho_i{(n)}}-\text{Pr}[N=n]\big)^2\Big] \label{eq:indepent-card-est-2} \\
  & = \sum\limits_{i \in {\mathcal{I}}}w_i^2 \text{MSE}_i^{(n)} \label{eq:MSE-AA-cardinality}
\end{align}
where $\text{Pr}[N=n]$ denotes the probability of $n$ targets in total. 
Then, similar result as \eqref{eq:MSEconsistency} can be obtained.
\end{proof}

\begin{remark} \label{remark:cardinality-wise-AA}
In addition to the consistency in terms of the estimator accuracy, the AA fusion can reduce the effects of clutter and missed detections \cite{Li17PC,Li17PCsmc} since the mixing operation and averaging calculation are insensitive to false and missing data, improving the robustness of the local filters. In the case of low detection profile, the AA fusion has a significant advantage as compared with the Bayes/multiplication-based fusion approaches. This, however, cannot be explained by the BFoM/MIL formulation. 
Simply, the GA fusion that minimizes similar divergence \eqref{eq:GA-density-KLD} has, however, been proven vulnerable to local misdetection \cite{Yu16,Li17merging,Li17PC,Li17PCsmc,Gao20cphd,Gao20GLMB,Da_Li20TSIPN,GLi22aaPMB} and may even underestimate the number of targets \cite{Gunay16,Li19Second,Uney19consistency}. 
\end{remark}

\subsection{PHD Filter Fusion} \label{sec:phd-AA-fusion}
By assuming a Poisson RFS formulation of the multi-target states as in \eqref{eq:Poisson-PD} where the Poisson parameter $\lambda =\int_{\mathcal{X}} D(\mathbf{x})d \mathbf{x}$, the PHD filter \cite{Mahler03,Vo06} 
recursively propagates the PHD in lieu of the Poisson MPD itself. As such, the PHD-AA fusion is very appropriate for the PHD filter fusion, where the resulting MPD is uniquely determined by the AA-fused PHD: 
\begin{equation}\label{eq:PHD-AA-PPP}
f_{\text{PHD-AA}}^{\text{p}}(X) = e^{-\int_{\mathcal{X}} D_{\text{AA}}(\mathbf{x})d \mathbf{x}} \prod_{\mathbf{x}\in X}D_{\text{AA}}(\mathbf{x})
\end{equation}
%
%
%
%
%
The PHD-AA fusion \eqref{eq:def-PHD-AA} for the PHD filter is consistent with the IIDC-AA fusion given in Lemma \ref{lemma:IIDC-AA}, which also \revisNew{amounts to} 
averaging 
the Poisson parameters, i.e., 
\begin{align}
  \lambda_{\text{PHD-AA}} & = \int_{\mathcal{X}} D_\text{AA}(\mathbf{x})d \mathbf{x} \nonumber  \\
   & = \sum\limits_{i \in {\mathcal{I}}} {w_i\lambda_i} \label{eq:phd-lambda-aa}
\end{align}

The PHD filter is actually the first testbed for applying the AA fusion in the context of multi-sensor RFS density fusion. A pioneering attempt \revisNew{of adopting averaging-for-fusion in the context of RFS filters was proposed with regard to the multi-target likelihood \cite{Streit08} which actually leads to the same result as averaging the posteriors if the priors are the same. }
Based on the average consensus approach and distributed peer-to-peer networks \cite{Li19Second,Li2021SomeResults}, the PHD-AA fusion has been proposed in \cite{Li17merging,Li17PC} and \cite{Li17PCsmc,Li19Diffusion} based on GM and particles implementation, respectively. Greatly, both can be viewed as a mixture distribution and 
allows greatly for closed-form
GM-PHD-AA and particle-PHD-AA fusion, respectively.
For extremely efficient communication and fusion-computation, one may resort to the cardinality-only fusion \cite{Li19CC} which shares and fuses only the cardinality (variables in the case of PHD filters and functions in the case of CPHD filters).

Following the same line of thinking, the CPHD filter \cite{Mahler07cphd, Vo07cphd} that recursively propagates the PHD and cardinality distribution can be fused in means of cardinality and PHD fusion, respectively \cite{Yu16,DaKai_Li_DCAI19}. This coincides with the result of the IIDC-AA fusion given in Lemma \ref{lemma:IIDC-AA}. 
Further issues such as limited field of view and asynchronous sensing have been addressed in \cite{Li19Diffusion,Yi20AAfov} and in \cite{GLi20AA-asynchronousPHD,Yu22asynchronousCPHD}, respectively. 

\subsection{IIDC PHD-AA Fusion}
The IIDC \eqref{eq:IIDC-p} point process which is the backbone of the CPHD filter \cite{Mahler07cphd,Vo07cphd}
is completely specified by the cardinality distribution $p(n)$ and the SPD $s(\mathbf{x})$.
We have the following important IIDC PHD-AA fusion result
\begin{lemma} \label{lemma:IIDC-AA}
For IIDC MPDs $f_{i}^\text{iidc}({X}_n) = n! \rho_i(n) \prod_{l=1}^{n} s_i(\mathbf{x})$, 
$i\in \mathcal{I}$, 
the cardinality-AA and PHD-AA fusion \eqref{eq:def-PHD-AA} using fusion weights $\mathbf{w} \in \mathbb{W}$ result in
\begin{equation}\label{eq:iidc-phd-aa-in-full}
  f_\text{PHD-AA}^\text{iidc}({X}_n) = n! \rho_\text{AA}(n) \prod_{l=1}^{n} s_\text{PHD-AA}(\mathbf{x})
\end{equation}
with AA fused cardinality distribution $\rho_{\text{AA}}(n)$ as given in \eqref{eq:def-rho-AA},
and the cardinalized-AA fused SPD as follows
\begin{align} %
s_\text{PHD-AA}(\mathbf{x})  
& = \frac{1}{\hat{N}_\text{AA}} \sum\limits_{i \in {\mathcal{I}}} {w_i \hat{N}_i s_i(\mathbf{x})} \label{eq:iidc-p-kld-s-gao}
\end{align}
where $\hat{N}_i \triangleq \int_{\mathcal{X}} {D_i(\mathbf{x})}  d \mathbf{x} $ 
and $\hat{N}_\text{AA} \triangleq \sum\limits_{i \in {\mathcal{I}}} {w_i \hat{N}_i}$. 
%
\end{lemma}

\begin{proof} The cardinality-AA results in straightforwardly the AA fused cardinality distribution $\rho_{\text{AA}}(n)$.
The PHD-AA fusion of the IIDCs can be expressed as follows
\begin{align}
D_{\text{AA}}^\text{iidc}(\mathbf{x}) & = \sum_{i \in \mathcal{I}} {w_i D_i^\text{iidc}(\mathbf{x})}  \label{eq:IIDC-PHD-AA} 
\end{align}
\revisNew{Since} the SPD is the normalized PHD, one can easily get \eqref{eq:iidc-p-kld-s-gao} from \eqref{eq:IIDC-PHD-AA}, i.e.,
\begin{align}
s_{\text{PHD-AA}}(\mathbf{x}) & = \frac{D_{\text{AA}}^\text{iidc}(\mathbf{x})}{\int_{\mathcal{X}} D_{\text{AA}}^\text{iidc}(\mathbf{x}) d \mathbf{x}} \nonumber \\
& = \frac{\sum_{i \in \mathcal{I}} {w_i D_i(\mathbf{x})} }{\sum_{i \in \mathcal{I}}w_i \int_{\mathcal{X}} {D_i(\mathbf{x})}  d \mathbf{x}} \nonumber \\
& = \frac{1}{\hat{N}_\text{AA}} \sum\limits_{i \in {\mathcal{I}}} {w_i \hat{N}_i s_i(\mathbf{x})}
\end{align}
where as defined in \eqref{eq:def-phd-integral=hatN}, $\hat{N}_i = \int_{\mathcal{X}} {D_i(\mathbf{x})}  d \mathbf{x}$.
\end{proof}

Note that the AA fused cardinality and PHD for IIDC fusion were first proposed in \cite[Theorem 4]{DaKai_Li_DCAI19}. The approximate BFoM/MIL derivation of the fused cardinality \eqref{eq:def-rho-AA} and fused SPD \eqref{eq:iidc-p-kld-s-gao} can be found in \cite[Proposition 1]{Gao20cphd}. Comparably, our derivation based on the PHD consistency imposes no \revisNew{ad-hoc} approximation and has a both theoretically and intuitively sound interpretation: the PHD-AA fusion leads to more accurate and robust detection of the present targets.

\section{Unlabeled Bernoulli/MB/MBM PHD-AA fusion} \label{sec:BMs} 
The Bernoulli distribution is conjugate prior for single-target filtering which suits 
for detecting and tracking at most one target in the presence of clutter and misdetection. In fact, the Bernoulli RFS has been an important cornerstone for developing many multi-target filters/trackers which can be briefly classified into two groups: One is the mixture such as the MBM \cite{Angel18PMBMdeivation} which admits closed-form AA fusion and the others such as the MB \cite{Vo09CBmember} do not. \revisNew{Loosely speaking,} when the BCs are labeled, they lead to ($\delta$-) GLMB \cite{Vo13Label,Vo14GLMB} and its simplified versions including the M-GLMB \cite{Fantacci15M-glmb}, LMB \cite{Reuter14LMB}, respectively. When an extra Poisson component is used (for representing the locally undetected targets), they lead to the PMBM \cite{Williams15taesPMBM} and PMB \cite{Williams15TSPpmb}. When both label and Poisson components are used, it leads to Poisson LMB \cite{Cament20PLMB}. 

The fusion of these mixtures {needs} to follow the first principle of \textit{target-wise fusion}\cite{Li20AAmb,Marano22}, i.e., only the information regarding the same, at least highly potentially the same, target should be fused. 
It is simply technically nonsensical to fuse the components that correspond to different targets. This target-wise fusion rule plays a key role in producing more accurate localization of the present targets. 
To comply with this rule, the Poisson component in the PMBM/PMB is not fused \cite{Li23AApmbm,GLi22aaPMB} since the locally undetected targets are by no means \revisNew{guaranteed to be} the same across sensors. 
For the remaining MBM/MB, BC to BC association and then \revisNew{merging} are the key in accurately localizing each detected target. 
%


\subsection{Bernoulli PHD-AA Fusion}

\begin{corollary} \label{corollary:Bernoulli}
For \revisNew{a group of} Bernoulli densities  $f^{\text{b}}_i(X)$ \revisNew{parameterized by} $\big(r_i, s_i(\mathbf{x})\big)$, $i \in \mathcal{I}$, 
the PHD-AA fusion \eqref{eq:def-PHD-AA} using fusion weights $\mathbf{w} \in \mathbb{W}$ will result in
\begin{align}\label{eq:iidc-phd-aa}
  f_\text{PHD-AA}^\text{b}({X}) = \begin{cases}
1-r_{\text{PHD-AA}} & X  =\emptyset\\
r_{\text{PHD-AA}}s_\text{PHD-AA}\left(\mathbf{x}\right) & X =\left\{ \mathbf{x}\right\} \\
0 & \mathrm{otherwise}
\end{cases}
\end{align}
where 
\begin{align}
  r_\text{PHD-AA}  & = \sum\limits_{i \in {\mathcal{I}}} {{w_i} {r_i}} \\
  s_\text{PHD-AA}(\mathbf{x}) & = \frac{1}{r_{\text{PHD-AA}}}{\sum\limits_{i \in {\mathcal{I}}} {{w_i}{r_i}{s_i}(\mathbf{x})} }
\end{align}
\end{corollary}

\begin{proof} Using \eqref{eq:Bernoulli-PHD}, the Bernoulli PHD-AA is
\begin{equation}\label{eq:Bernoulli-phd-aa}
D^\text{b}_{\text{AA}}(\mathbf{x})  = r_\text{PHD-AA} s_\text{PHD-AA}(\mathbf{x}) 
\end{equation}
where $r_\text{PHD-AA}$ and $s_\text{PHD-AA}(\mathbf{x})$ are so defined to ensure PHD consistency by $D^\text{b}_{\text{AA}}(\mathbf{x})  = \sum\limits_{i \in {\mathcal{I}}} {w_i} D^\text{b}_i(\mathbf{x}) = \sum\limits_{i \in {\mathcal{I}}} {w_i}r_i s_i(\mathbf{x})$. So, 

\begin{align}
r_\text{PHD-AA} & = \int_{\mathcal{X}} D^\text{b}_{\text{AA}}(\mathbf{x}) d \mathbf{x} \nonumber \\
& = \int_{\mathcal{X}} \sum\limits_{i \in {\mathcal{I}}} {w_i} D^\text{b}_i(\mathbf{x})d \mathbf{x} \nonumber \\
& = \sum\limits_{i \in {\mathcal{I}}} {{w_i} {r_i}}   
\end{align}
\begin{align}
s_\text{PHD-AA} & = \frac{1}{r_\text{PHD-AA}}D^\text{b}_{\text{AA}}(\mathbf{x}) \nonumber \\
& = \frac{1}{r_{\text{PHD-AA}}}{\sum\limits_{i \in {\mathcal{I}}} {{w_i}{r_i}{s_i}(\mathbf{x})} }  
\end{align}
In other words, the only BC that ensures PHD-AA fusion is given by the above Bernoulli-AA fusion.
\end{proof}

The Bernoulli process with parameters $\big(r, s(\mathbf{x})\big)$ is a special IIDC process with $\rho(0)=1-r, \rho(1)=r$ and $\rho(n>1)=0$. The above results can be viewed as a special case of Lemma \ref{lemma:IIDC-AA} and coincide with the following MPD-AA results given in \cite[Theorem 1]{Li19Bernoulli} and \cite[Lemma 1]{DaKai_Li_DCAI19}.
\begin{align}\label{eq:iidc-mpd-aa}
  f_\text{PHD-AA}^\text{b}({X}) =  \sum\limits_{i \in {\mathcal{I}}} {w_i}f_i^\text{b}({X})
\end{align}


%
%
%


\subsection{MB PHD-AA Fusion} \label{sec:MB-AA}
Extending the single-target-BC to the multi-target RFS, MB lies in the core of MBM \cite{Angel18PMBMdeivation}, ($\delta$-) GLMB \cite{Vo13Label,Vo14GLMB}, PMB \cite{Williams15TSPpmb}, M-GLMB \cite{Fantacci15M-glmb}, labeled MB (LMB) \cite{Reuter14LMB} and Poisson LMB \cite{Cament20PLMB}.
Obviously, the AA fusion of the MBs will lead to a mixture of MBs \cite{Li20AAmb} which should be reduced to a single MB by merging all those BCs corresponding to the same target. 
That is 
to carry out the AA fusion in parallel to these associated/matched BCs from different sensors, for which {the} similarity or divergence between these components plays a key role in evaluating whether the components are close enough and so be merged/fused to one.

Here, we consider the MB-AA fusion from the perspective of PHD-AA/consistency. 
Consider the MB RFS which has $n_i$ BCs in which the $j$-th BC has existence probability $r_{i,j}$ and SPD $s_{i,j}\left(\mathbf{x}\right)$. \revisNew{A 1D example consisting of two BCs represented by three SPDs associated with different existing probabilities is given in Fig. \ref{fig:PHDofMB}.}
 The corresponding PHD-AA for the MBs can be expressed as, c.f., \eqref{eq:def-PHD-AA},
\begin{align}
	D_{\text{AA}}^{\text {mb}}(\mathbf{x}) & =  \sum\limits_{i \in {\mathcal{I}}}{w_i D_i^{\text {mb}}(\mathbf{x}) } \nonumber \\
& =  \sum\limits_{i \in {\mathcal{I}}}{w_i  \sum_{j=1}^{n_i}{r_{i,j}}{s_{i,j}}(\mathbf{x}) } \nonumber \\ %
& =  \sum_{j=1}^{n_{f}}{r_{j,\text{PHD-AA}}}{s_{j,\text{PHD-AA}}}(\mathbf{x})  \label{eq:MB-PHA-AA}
\end{align}
which corresponds to the following fused MB MPD, c.f., \eqref{eq:def_MB},
\begin{equation}\label{MB-AA-MPD}
  	f_{\text{PHD-AA}}^\text{mb} (X_n)  = \sum_{\uplus_{j=1}^{n_{f}}X^{(j)}=X}\prod_{j=1}^{n_{f}}f_{\text{PHD-AA}}^\text{b} \left(X^{(j)}\right)
\end{equation}
where \revisNew{each multiple-to-one} fused BC $f_{\text{PHD-AA}}^\text{b} \left(X^{(j)}\right)$ with parameters $\big( r_{j,\text{PHD-AA}},{s_{j,\text{PHD-AA}}}(\mathbf{x}) \big)$ is given in \eqref{eq:iidc-phd-aa} in Corollary \ref{corollary:Bernoulli} for each clustered/associated group of BCs, and the final number $n_{f}$ of fused BCs depends on the BC merging/clustering procedure, usually
\begin{equation}\label{eq:MB-n}
  \max\{n_i\}_{i \in {\mathcal{I}}} \leq n_{f} \leq  \sum_{i \in {\mathcal{I}}}n_i
\end{equation}
where the left equation holds iff all BCs from the other sensors are associated with those of a local sensor and the right equation holds iff no BCs are associated with \revisNew{the others}.

\begin{figure}
  \centering
  \revisNew{
  \includegraphics[width=6.5cm]{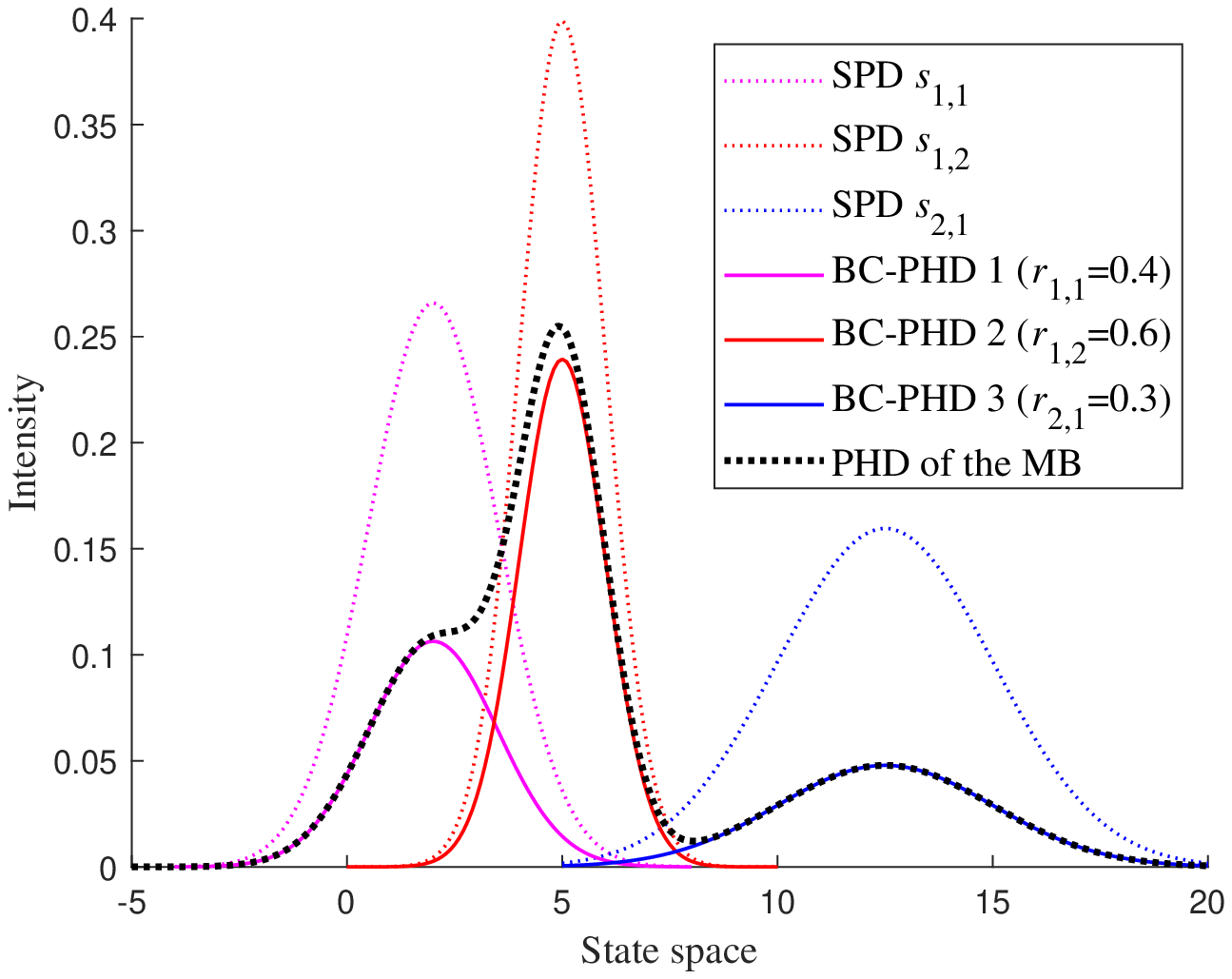}\\
  \caption{An MB density consisting of two BCs that are represented by three SPDs with different existing probabilities. }\label{fig:PHDofMB}
  \vspace{-3mm}
  }
\end{figure}

The readers are referred to \cite{Li20AAmb} for the detail of grouping MBs via clustering or pairwise association \cite{Li20AAmb}. \revisNew{For this purpose}, data-driven clustering and Bernoulli merging \cite{Fontana22merging} may be useful. The above MB-AA fusion has been used for PMB-AA-fusion \cite{GLi22aaPMB}. In the case of labeled MB, the BC association amounts to the label matching \revisNew{as to be addressed in section \ref{sec:label}}. 

\subsection{MBM PHD-AA Fusion} \label{sec:MBM-AA}
Since the MBM is a mixture that admits closed-form AA fusion \cite{Li23AApmbm}, the AA fusion is the perfect match to them.
Consider the MBM RFS which has $\mathbb{J}_i$ MBs, each having $n_{i,j}$ BCs in which the $l$-th BC has existence probability $r_{i,j,l}$ and SPD $s_{i,j,l}\left(\mathbf{x}\right)$. The corresponding PHD-AA for the MBMs can be expressed as, c.f., \eqref{eq:def_MBM_PHD},
\begin{align}
	D_{\text{AA}}^{\text {mbm}}(\mathbf{x}) & =  \sum\limits_{i \in {\mathcal{I}}}{w_i D_i^{\text {mbm}}(\mathbf{x}) } \nonumber \\
& =  \sum\limits_{i \in {\mathcal{I}}}{w_i \sum\limits_{j \in \mathbb{J}_i} \sum_{l=1}^{n_{i,j}}{r_{i,j,l}}{s_{i,j,l}}(\mathbf{x}) } \nonumber \\
& =  \sum\limits_{(i,j) \in \mathcal{I} \times \mathbb{J}_i} \sum_{l=1}^{n_{i,j}}{r'_{i,j,l}}{s_{i,j,l}(\mathbf{x})} \label{eq:MBM-PHA-AA}
\end{align}
where the re-weighted existence probability of each BC is
\begin{equation}\label{hypothesis_Union}
  	r'_{i,j,l} = w_ir_{i,j,l}
\end{equation}

It can be seen that the number of MBs in the fused MBM equals the sum of those of the fusing MBMs. To combat with the increase of the MBM size, it is not recommended to communicate and fuse the entire MBM, which is not only communicatively and computationally costly but may result in lower accuracy according to our experience. It is just sufficient to communicate and fuse only a few MBs or even a single MB \revisNew{of the greatest hypothesis weight} in each fusing MBM \cite{Li23AApmbm}. Furthermore, \revisNew{the BC association and merging should be carried out following the \textit{target-wise} fusion rule: the BCs regarding the same potential target should be merged, although the data association is usually challenging \cite{Marano22}.} Moreover, the BC merging is supposed to be carried out to those from different MBs of different sensors (as addressed in \cite[Sec.5.2]{Li23AApmbm}) while the BCs in the same MB are supposed to represent different targets and should not be fused.

%
%


\begin{remark}
\revisNew{BC association and merging are the key operations to execute the target-wise AA fusion of these BC-mixture type filters including MB, PMB, MBM, PMBM (and even the GLMB in spite of the extra labels)} and their simplified variants.
Based on proper inter-sensor track/BC association, these filters can all be fused in a similar way that facilitates parallel, closed-form Bernoulli merging. \revisNew{By this, the resulted BC can be more concentrated on its mode/mean. } 
\end{remark}


\section{Labeled PHD-AA fusion} \label{sec:LRFS-AA-fusion} \label{sec:label}
The LRFS filters such as the GLMB filters \cite{Vo13Label,Vo14GLMB} provide 
the time-series trajectory information for each target via the labels, which makes them advantageous in comparison with the unlabeled RFS filters by eliminating problems like track fragmentation \cite{Vo19msGLMB}. 
However, just because of the use of labels that contain track identification information, 
the fusion becomes more challenging as it is needed to ``\textit{coordinate/match}'' different labels/tracks among the fusing LRFS filters.
\revisNew{
\begin{definition}[Label matching] The labels assigned to the BCs that represent the same target should be matched with each other: 
Only the matched components be fused and be relabeled the same after the fusion. Two components are deemed irrelevant and should not be involved in the fusion if they are unmatched.
\end{definition}}

This \textit{label-wise} fusion can be viewed as a strict implementation of the \textit{target-wise fusion} rule, which is the prerequisite for the LRFS fusion.
\revisNew{Unlike in the unlabeled case where 
the fusing components that correspond to the same target may not be merged if they are distant, in the case of LRFS fusion all of the labeled components corresponding to the same target must be matched and relabeled using the same, unique label. This is simply because a single target should not have multiple labels due to the {label uniqueness}. 
By label matching, it simply means that one can be converted to another through a bijective mapping function.} 
The problem of \revisNew{the so-called} label inconsistency and matching has been addressed earlier \cite{SLi18LableInconsistenceFreeFusion,Sli19GCI-matching-GLMB,Gao20GLMB}, which remains an open issue. \revisNew{In this work, we assume that labels can be properly matched across the sensors and leave how to match labels to the future work}.


%

\subsection{Divergence between GLMBs/LPHDs} \label{sec:LRFS-label}

The LRFS integral is defined as \cite{Vo13Label,Vo14GLMB,Papi15}
\begin{align}
 \int_{\mathbb{X} \times \mathbb{L}} \pi\big(\widetilde{X}\big) \delta \widetilde{X} 
 & =  \sum\limits_{n = 0}^\infty \sum_{L_n \subseteq  \mathcal{L}^n} \frac{1}{{n!}} \int_{\mathcal{X}^n} \pi\big(\widetilde{X}_n\big) d \mathbf{x}_1 ,\dots,d \mathbf{x}_n \nonumber
\end{align}

{In the above formulation, the labels are random, discrete variables. However, in the context of labeled MPD fusion, the labels of any \revisNew{realization of the} labeled MPD are deterministic and finite, which need to be matched between the fusing labeled filters to abide with the label-wise fusion rule.} 
\revisNew{The existing definition of the KL divergence of two LRFS densities with the same given} label set $L_n=\{l_1,l_2,\dots,l_n\}$ 
is
\begin{align}
 D_\text{KL}& \big( \pi_1^{L_n} \| \pi_2^{L_n}  \big) = \int_{\mathcal{X}^n}  {\pi_1\big(\widetilde{X}_n\big)\log \frac{\pi_1\big(\widetilde{X}_n\big)}{\pi_2\big(\widetilde{X}_n\big)} \delta X_n} \label{eq:def_LRFS-kld}  
\end{align}

\begin{remark}
The above definition is restrictive as it holds only for two labeled MPDs \revisNew{realized with the same number of labels}, where 
the inter-filter labels are matched and treated as the same. It takes into account only the divergence in the state distribution and \revisNew{ignores the difference between the matched labels. This amounts to completely ignoring the label estimate error of each filter}. In principle, the labels which unavoidably bear estimate error are variant among the fusing filters and their difference should be accounted for in the divergence for labeled MPDs \revisNew{otherwise the result as calculated by \eqref{eq:def_LRFS-kld} is incomplete and problematic.}
\footnote{Considering two labeled MPDs of the same state distribution but assigned with label $(1; 1)$ and $(2; 1)$, respectively. Here, in the label notation $l = (k; \kappa)$, $k$ denotes the estimated time of birth, and $\kappa$ is a unique index to distinguish new targets born at the same time \cite{Vo13Label}. That is, both filters detect a target and produce the same state distribution estimate but different born-times of the track.  By matching these two labels, the divergence calculated by \eqref{eq:def_LRFS-kld} will be zero which means erroneously that two labeled MPDs are identical. 
If these two labels are not matched, \eqref{eq:def_LRFS-kld} will yield \revisNew{no meaningful result.} 
The result is unconvincing in either case.
} 
%
\end{remark}

\begin{figure}
  \centering
  \includegraphics[width=8.5cm]{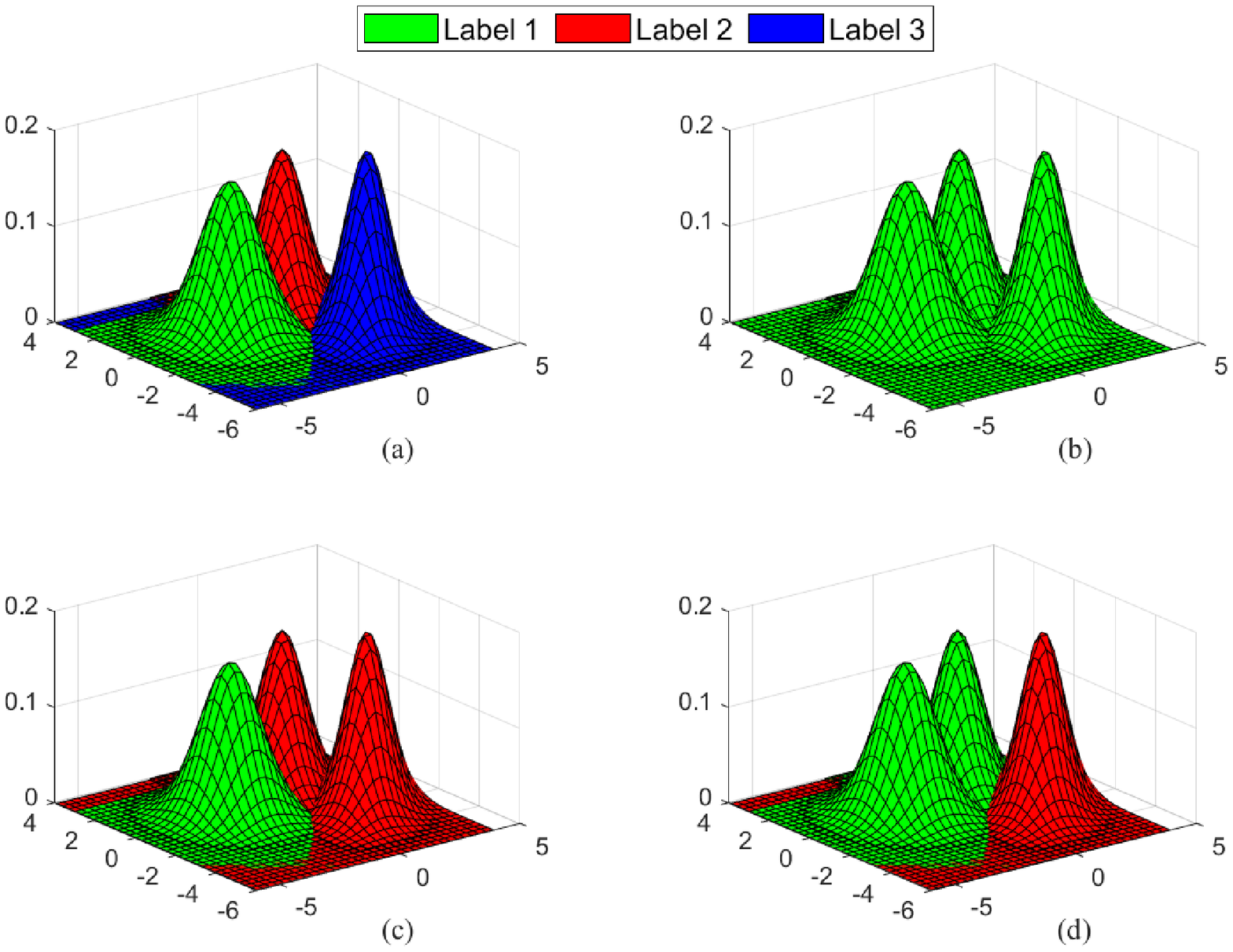}\\  
  \caption{Four LMBs of the same state distribution but different labels: \revisNew{The so-defined KL divergence \eqref{eq:def_LRFS-kld} applies only between the two LRFS distributions (c) and (d) that have the same number of labels, not the other cases that have different numbers of labels. 
  }  
   } \label{fig:LMB-4}
  \vspace{-2mm}
\end{figure}

This can be illustrated in Fig. \ref{fig:LMB-4} where four LMBs have the same unlabeled state distribution but different labels and different potential target-state-estimates. 
\revisNew{The divergence \eqref{eq:def_LRFS-kld} applies only between the LRFS distributions of the same number of labels and different label matching choices may result in different results. However, it remains undefined how to calculate the divergence between any other two labeled distributions that have different numbers of labels such as the one given in (a) and (b)/(c)/(d) in Fig. \ref{fig:LMB-4}. This leads to a theoretical barrier to derive the MPD-AA fusion based on the BFoM \eqref{eq:RFS-AA-Whole-KLD}.}
This problem, however, does not exist with the LPHD \eqref{eq:def-LRFS-phd-label} that is rigorously defined on single-target single-label $\mathcal{X} \times \mathcal{L}$, no matter how many labels each filter has. 

We therefore extend the definition of the PHD consistency to the labeled domain and \revisNew{assume the LPHDs to be fused are properly label-matched (and so the matched components will be relabeled the same after fusion). That is, label matching is the prerequisite of the state density fusion.
We reiterate that} the LPHD-AA fusion is a necessary but not sufficient condition for the MPD-AA fusion, namely, 
\begin{align}
  {\widetilde D}_{\text {MPD-AA}} (\mathbf{x},l) 
  & = \int_{\mathbb{X} \times \mathbb{L}}  {\bigg(\sum_{\mathbf{w}\in {\widetilde X}}{\delta_\mathbf{w}}\big[(\mathbf{x},l)\big] \bigg)\pi_{\text {AA}}({\widetilde X})\delta {\widetilde X}}  \nonumber \\
  & = \sum_{i \in  \mathcal{I}}  {w_i}\int_{\mathbb{X} \times \mathbb{L}}   {\bigg(\sum_{\mathbf{w}\in {\widetilde X}}{\delta_\mathbf{w}}\big[(\mathbf{x},l)\big] \bigg)\pi_i({\widetilde X})\delta {\widetilde X}}  \nonumber \\
  & = \sum_{i \in \mathcal{I}} {w_i {\widetilde D}_i(\mathbf{x},l)}
  \nonumber \\
  & \triangleq {\widetilde D}_{\text {AA}}(\mathbf{x},l) \label{eq:GLMB-aa-phd-label}
\end{align}
which has the following BFoM property, c.f., \eqref{eq:PHD-AA-KLD}
\begin{align}
 {\widetilde D}_{\text{AA}}({\mathbf{x}},l) & = \operatorname*{arg\,min}_{g\in \mathcal{F}_{\mathcal{X}\mathcal{L}}} \sum\limits_{i \in {\mathcal{I}}} {w_i D_\text{KL}\big({\widetilde D}_i \|g\big)} \label{eq:LPHD-AA-KLD}
\end{align}
where 
the scalar-valued function set $\mathcal{F}_{\mathcal{X}\mathcal{L}} = \{f: \mathcal{X} \times \mathcal{L} \rightarrow \mathbb{R} \}$, and the KL divergence \revisNew{with respect to two LPHDs ${\widetilde D}_{1}({\mathbf{x}},l),{\widetilde D}_{2}({\mathbf{x}},l)$ of the same label $l \in \mathcal{L}$ is defined as}
\begin{align}
 D_\text{KL} \big( {\widetilde D}_1 \| {\widetilde D}_2 \big) = 
 \int_{\mathcal{X}}  {{\widetilde D}_1 ({\mathbf{x}},l) \log \frac{{\widetilde D}_1({\mathbf{x}},l)}{{\widetilde D}_2({\mathbf{x}},l)} d \mathbf{x}} \label{eq:def_LPHD-kld}  
\end{align}  %

\revisNew{We hereafter start from averaging the state distribution of the LPHDs with matched label as shown in \eqref{eq:GLMB-aa-phd-label} but do not calculate  \eqref{eq:LPHD-AA-KLD} or  \eqref{eq:def_LPHD-kld} which overlooks the label estimate and matching error. It is the LPHD consistency not the BFoM that forms the motivation of our derivation. }

\subsection{Index-matched GLMB LPHD-AA Fusion} \label{sec:label-wise-fusion-base}
\begin{lemma} \label{lemma:glmb-lphd-aa}
\revisNew{For a group of GLMBs $\pi_i^{\text {gl}}\big(\widetilde{X}\big), i \in {\mathcal{I}}$ with matched index $c \in \mathbb{C}$, defined as follows, c.f., \eqref{eq:glmb-in-general}
\begin{equation}
	\pi_i^{\text {gl}}\big(\widetilde{X}\big) =  \Delta \big(\widetilde{X}\big) \sum\limits_{c \in \mathbb{C}} {{\omega_i^{(c)}}\Big(\mathcal{L}\big(\widetilde{X}\big)\Big)\prod\limits_{(\mathbf{x},l) \in \widetilde{X}} {{s_i^{(c)}}(\mathbf{x},l)} }
\end{equation}
the corresponding LPHD-AA fusion using fusion weights $\mathbf{w} \in \mathbb{W}$ will result in the following fused GLMB
\begin{align}
	& \pi_\text{LPHD-AA}^{\text {gl}}\big(\widetilde{X}\big) =  \nonumber \\
& \Delta \big(\widetilde{X}\big) \sum\limits_{c \in \mathbb{C}} {{\omega_\text{LPHD-AA}^{(c)}}\Big(\mathcal{L}\big(\widetilde{X}\big)\Big)\prod\limits_{(\mathbf{x},l) \in \widetilde{X}} {{s_\text{LPHD-AA}^{(c)}}(\mathbf{x},l)} }
\end{align}
where}
\begin{align}
 {\omega_\text{LPHD-AA}^{(c)}}(L) & =  \sum\limits_{i \in {\mathcal{I}}} {{w_i}{\omega_i^{(c)}}(L)} \label{eq:def-glmb-lphd-aa-omega} \\
  s_\text{LPHD-AA}^{(c)}(\mathbf{x},l) & = \frac{1}{ {\omega_\text{LPHD-AA}^{(c)}}(L) } \sum_{i \in \mathcal{I}} w_i {\omega_i^{(c)}(L)} {s_i^{(c)}}(\mathbf{x},l) \label{eq:def-glmb-lphd-aa-s}
\end{align}
\end{lemma}
\begin{proof}
\revisNew{We first calculate the LPHD ${{\widetilde D}^{\text {gl}}_i}(\mathbf{x},l)$ of each GLMB, obtaining $\forall i \in {\mathcal{I}}$
\begin{equation}
  {\widetilde D}_i^{\text {gl}}(\mathbf{x},l)  = \sum\limits_{c \in \mathbb{C}}{{s_i^{(c)}}(\mathbf{x},l)}\sum\limits_{L \subseteq \mathbb{L}} {{\mathrm{1}_L}(l) {\omega_i^{(c)}}(L)}
\end{equation}
Given the matched index $c \in \mathbb{C}$, their AA using fusion weights $\mathbf{w} \in \mathbb{W}$ is given by
\begin{align}
{\widetilde D}^{\text {gl}}_\text{AA} (\mathbf{x},l) & = \sum_{i \in \mathcal{I}} {w_i {\widetilde D}^{\text {gl}}_i (\mathbf{x},l)} \label{eq:glmb-Lphd-aa} \\
& = \sum\limits_{c \in \mathbb{C}}{{s_\text{LPHD-AA}^{(c)}}(\mathbf{x},l)}\sum\limits_{L \subseteq \mathbb{L}} {{\mathrm{1}_L}(l) {\omega_\text{LPHD-AA}^{(c)}}(L)} \label{eq:def-glmb-lphd-aa}
\end{align}}
Calculating the state integral of both sides of \eqref{eq:glmb-Lphd-aa} leads to
$ \int_{\mathcal{X}} { \sum\limits_{c \in \mathbb{C}} {{s_\text{LPHD-AA}^{(c)}}(\mathbf{x},l)} } \sum\limits_{L \subseteq \mathbb{L}}{\mathrm{1}_L}(l){ \omega_\text{LPHD-AA}^{(c)}(L)} d \mathbf{x} 
   = \int_{\mathcal{X}} \sum_{i \in \mathcal{I}} w_i {\sum\limits_{c \in \mathbb{C}}{s_i^{(c)}}(\mathbf{x},l)} \sum\limits_{L \subseteq \mathbb{L}}{\mathrm{1}_L}(l) { {\omega_i^{(c)}(L)} } d \mathbf{x}$.
That is,
\begin{align} 
 \sum\limits_{L \subseteq \mathbb{L}} {\mathrm{1}_L}(l) \sum\limits_{c \in \mathbb{C}} { \omega_\text{LPHD-AA}^{(c)}(L)} = \sum\limits_{L \subseteq \mathbb{L}}{\mathrm{1}_L}(l)  \sum\limits_{c \in \mathbb{C}} \sum_{i \in \mathcal{I}} w_i  { \omega_i^{(c)}(L)} \nonumber
\end{align}

Given that the index $c \in \mathbb{C}$, label set $L \subseteq \mathbb{L}$ are matched among sensors, respectively, one can easily get \eqref{eq:def-glmb-lphd-aa-omega} for $L: {\mathrm{1}_L}(l)=1$.
Meanwhile, \eqref{eq:glmb-Lphd-aa} can be expressed as follows
\begin{align}
  & \sum\limits_{L \subseteq \mathbb{L}} { \sum\limits_{c \in \mathbb{C}} {{\mathrm{1}_L}(l) \omega_\text{LPHD-AA}^{(c)}(L)} {{s_\text{LPHD-AA}^{(c)}}(\mathbf{x},l)} } \nonumber \\
  & = \sum_{i \in \mathcal{I}} w_i \sum\limits_{L \subseteq \mathbb{L}} { \sum\limits_{c \in \mathbb{C}} {{\mathrm{1}_L}(l) \omega_i^{(c)}(L)} {{s_i^{(c)}}(\mathbf{x},l)} } \label{eq:full-GLMB-LPHD-aa}
\end{align}
which \revisNew{leads to \eqref{eq:def-glmb-lphd-aa-s} for} each matched index $c$ and label $l$. 
\end{proof}

\revisNew{\begin{remark}
The fusion of the GLMBs in Lemma \ref{lemma:glmb-lphd-aa} 
is the LPHD-AA consistency as given in \eqref{eq:glmb-Lphd-aa}, not the LPHD-best-fit like \eqref{eq:LPHD-AA-KLD}. As addressed, the latter holds only in the case when the fusing LPHDs have the same label since the divergence \eqref{eq:def_LPHD-kld} 
is undefined for the LPHDs of different labels. 
In contrast, the LPHD-AA \eqref{eq:glmb-Lphd-aa} can still be calculated for LPHDs realized with actually different labels, ${\widetilde D}^{\text {gl}}_1 (\mathbf{x},l_i), i \in \mathcal{I}$, i.e.,
\begin{equation}\label{eq:LPHD-AA-diff-labels}
  {\widetilde D}^{\text {gl}}_\text{AA} (\mathbf{x},l_\text{f}) \triangleq \sum_{i \in \mathcal{I}} {w_i {\widetilde D}^{\text {gl}}_i (\mathbf{x},l_i)}
\end{equation}
where $l_\text{f}$ denotes the re-created label to relabel all components involved in the right side after fusion.
\end{remark}}

\subsection{$\delta$-GLMB LPHD-AA Fusion with Matched Label Set and Track Association Hypotheses}
Given that the labels $l\in \mathcal{L}$, as well as the relevant history data association hypothesis pair $(L,\xi) \subseteq \mathbb{L} \times \Xi$, are matched across the sensors, 
we have the following result



\begin{corollary} \label{corollary:delta-GLMB-AA}
For \revisNew{a group of} $\delta$-GLMBs $\pi_\text{AA}^{{\delta}}\big(\widetilde{X}\big), i\in \mathcal{I}$ with matched labels $l\in \mathcal{L}$ and track association hypotheses $(L,\xi) \subseteq \mathbb{L} \times \Xi$, \revisNew{defined as follows,}
$$\pi_i^{{\delta}}\big(\widetilde{X}\big) = \Delta \big(\widetilde{X}\big) \sum\limits_{L \subseteq \mathbb{L}} {\delta _L}\big[\mathcal{L}\big(\widetilde{X}\big)\big]  \sum\limits_{\xi \in \Xi} {{\omega_i^{(L,\xi )}}\prod\limits_{(\mathbf{x},l) \in \widetilde{X}} {{s_i^{(\xi)}}(\mathbf{x},l)}}$$ 
the corresponding LPHD-AA fusion using fusion weights $\mathbf{w} \in \mathbb{W}$ will result in the \revisNew{following} fused $\delta$-GLMB
\begin{align}
  & \pi_\text{LPHD-AA}^{{\delta}}\big(\widetilde{X}\big) = \nonumber \\
  & \Delta \big(\widetilde{X}\big) \sum\limits_{L \subseteq \mathbb{L}} {\delta _L}\big[\mathcal{L}\big(\widetilde{X}\big)\big]  \sum\limits_{\xi \in \Xi} {{\omega_\text{LPHD-AA}^{(L,\xi )}}\prod\limits_{(\mathbf{x},l) \in \widetilde{X}} {{s_\text{LPHD-AA}^{(\xi)}}(\mathbf{x},l)}}\label{eq:delta-GLMB-AA-in-full}
\end{align}
where 
 \begin{align}
  {\omega_\text{LPHD-AA}^{(L,\xi )}} & =  \sum_{i \in \mathcal{I}} w_i  {\omega_i^{(L,\xi )}} \label{eq:omega-AA-hypothesis-matching} \\
   {{s_\text{LPHD-AA}^{(\xi)}}(\mathbf{x},l)} & = \frac{1}{\sum_{i \in \mathcal{I}} w_i {\omega_i^{(L,\xi )}}} \sum_{i \in \mathcal{I}} w_i {\omega_i^{(L,\xi )}} {s_i^{(\xi)}}(\mathbf{x},l) \label{eq:s-AA-hypothesis-matching}
 \end{align}
\end{corollary}

\begin{proof}
The proof can be similarly done as that for Lemma \ref{lemma:glmb-lphd-aa} in view of the fact that a $\delta$-GLMB RFS is a special GLMB RFS with $\mathbb{C} =  \mathbb{L} \times \Xi$, ${\omega^{(c)}}(L') = {\omega^{(L,\xi )}}{\delta _L}(L')$, and ${s^{(c)}}(\mathbf{x},l) = {s^{(\xi )}}(\mathbf{x},l)$.
\end{proof}

\subsection{Simplified GLMB LPHD-AA Fusion}
Two computationally cheaper, approximate alternatives to the $\delta$-GLMB filter \cite{Vo13Label,Vo14GLMB} 
are given by the M-GLMB \cite{Fantacci15M-glmb} and the LMB \cite{Reuter14LMB}. The former preserves both the LPHD and cardinality distribution while the latter only matches the unlabeled PHD but not the cardinality distribution \cite{Fantacci18GCI-lmb-Mglmb}.

In short, the M-GLMB filter \cite{Fantacci15M-glmb} marginalizes the history measurement-track association hypothesis $\xi$ over $\Xi$ of the $\delta$-GLMB, namely,
\begin{equation}\label{eq:marg-M-GLMB}
  \omega^{(L)} = \sum\limits_{\xi \in \Xi} {\omega^{(L,\xi )}}
\end{equation}
which leads to the following MPD and LPHD
\begin{align}
	\pi^{\text{m-}\delta}\big(\widetilde{X}\big) &= \Delta \big(\widetilde{X}\big) \sum\limits_{L \subseteq \mathbb{L}} {\delta _L}\big[\mathcal{L}\big(\widetilde{X}\big)\big]  {{\omega^{(L)}}\prod\limits_{(\mathbf{x},l) \in \widetilde{X}} {{s^{(L)}}(\mathbf{x},l)} } \label{eq:M-GLMB} \\
{\widetilde D}^{\text{m-}\delta}(\mathbf{x},l) &= \sum\limits_{L \subseteq \mathbb{L}} {\mathrm{1}_L}(l)  {{\omega ^{(L)}} { {s^{(L)}}(\mathbf{x},l)}} \label{eq:M-GLMB-phd}
\end{align}

The LMB filter \cite{Reuter14LMB} approximates the GLMB using a single LMB distribution, namely $\mathbb{C}$ has only a single element in \eqref{eq:glmb-in-general}. 
Then, each labeled BC with label $l \in L \subseteq \mathbb{L}$ is associated with a SPD ${s}(\mathbf{x},l)$ and an existence probability $r^{(l)}$, given by
 \begin{align}
   r^{(l)} 
   & =  \sum\limits_{L \subseteq \mathbb{L}} {\mathrm{1}_L}(l) \sum\limits_{\xi \in \Xi} {\omega^{(L,\xi )}}\label{eq:LMB-l-EP} \\
   {s}(\mathbf{x},l) & = \sum\limits_{L \subseteq \mathbb{L}} {\mathrm{1}_L}(l) \sum\limits_{\xi \in \Xi} {\omega^{(L,\xi )}}{s^{(\xi)}}(\mathbf{x},l) \label{eq:LMB-l-s}
 \end{align}
The corresponding MPD and LPHD are given by
\begin{align}
	\pi^{\text {lmb}}(\widetilde{X}) &= \Delta \big(\widetilde{X}\big) \omega\big[\mathcal{L}\big(\widetilde{X}\big)\big] \prod\limits_{(\mathbf{x},l) \in \widetilde{X}} {{s}(\mathbf{x},l)} \label{eq:LMB} \\
{\widetilde D}^{\text {lmb}}(\mathbf{x},l) 
& = r^{(l)} {s}(\mathbf{x},l) \label{eq:LMB-phd} 
\end{align}


Both M-GLMB-LPHD and LMB-LPHD AA fusion can be derived from Corollary \ref{corollary:delta-GLMB-AA} by carrying forward the \revisNew{corresponding simplification that they made on} the $\delta$-GLMB. 
\begin{corollary} \label{corollary:M-GLMB-AA}
For \revisNew{a group of} M-GLMBs $\pi_i^{\text{m-}\delta}\big(\widetilde{X}\big), i\in \mathcal{I}$ with matched labels $l\in \mathcal{L}$  and label set $L \subseteq \mathbb{L}$,
$$\pi_i^{\text{m-}\delta}\big(\widetilde{X}\big) = \Delta \big(\widetilde{X}\big) \sum\limits_{L \subseteq \mathbb{L}} {\delta _L}\big[\mathcal{L}\big(\widetilde{X}\big)\big]  {{\omega_i^{(L)}}\prod\limits_{(\mathbf{x},l) \in \widetilde{X}} {{s_i^{(L)}}(\mathbf{x},l)} }$$ 
the corresponding LPHD-AA fusion using fusion weights $\mathbf{w} \in \mathbb{W}$ will result in the fused M-GLMB
$\pi_\text{LPHD-AA}^{\text{m-}\delta}\big(\widetilde{X}\big) =  \Delta \big(\widetilde{X}\big) \sum\limits_{L \subseteq \mathbb{L}} {\delta _L}\big[\mathcal{L}\big(\widetilde{X}\big)\big]  {{\omega_\text{LPHD-AA}^{(L)}}\prod\limits_{(\mathbf{x},l) \in \widetilde{X}} {{s_\text{LPHD-AA}^{(L)}}(\mathbf{x},l)} }$,
where, 
 \begin{align}
  {\omega_\text{LPHD-AA}^{(L)}} & =  \sum_{i \in \mathcal{I}} w_i  {\omega_i^{(L)}} \label{eq:omega-AA-hypothesis-matching-M-GLMB} \\
   {{s_\text{LPHD-AA}^{(\xi)}}(\mathbf{x},l)} & = \frac{1}{\omega_\text{LPHD-AA}^{(L)}} \sum_{i \in \mathcal{I}} w_i {\omega_i^{(L)}} {s_i^{(L)}}(\mathbf{x},l) \label{eq:s-AA-hypothesis-matching-M-GLMB}
 \end{align}
\end{corollary}

\begin{proof}
\revisNew{The result is} from Corollary \ref{corollary:delta-GLMB-AA} by using \eqref{eq:marg-M-GLMB} and replacing ${{s_\text{LPHD-AA}^{(\xi)}}(\mathbf{x},l)}$ by ${{s_\text{LPHD-AA}^{(L)}}(\mathbf{x},l)}$. 
It has been derived from the \revisNew{constrained} MPD-best-fit in \cite{Gao20GLMB} based on a compressed expression of the $\delta$-GLMB density \cite{Papi15}. 
\end{proof}


\begin{corollary} \label{corollary:LMB-AA}
For \revisNew{a group of LMBs consisting of labeled BCs} with existence probability $r_i^{(l)}$ and SPD ${s_i}(\mathbf{x}, l)$, $i\in \mathcal{I}$, {$l\in L$}, \revisNew{where $L$ denotes the common label set of all fusing LMBs}, the label-wise LPHD-AA fusion leads to a fused LMB with AA fused existence probability $r_{\text{LPHD-AA}}^{(l)}$ and cardinalized-AA fused SPD $s_{\text{LPHD-AA}}(\mathbf{x}, l)$ given as follows, $\forall l\in L$
\begin{align}
{r_{\text{{LPHD-AA}}}^{(l)}} & =  \sum_{i \in  \mathcal{I}}  {{w_i} {r_i}^{(l)}}\label{eq:L_RFS-AA-JEP} \\
s_{\text{{LPHD-AA}}}(\mathbf{x}, l) & =\frac{1}{\revisNew{r_{\text{LPHD-AA}}^{(l)}}}{\sum_{i \in \mathcal{I}}  {{w_i}{r_i^{(l)}} {s_i}(\mathbf{x}, l)} } \label{eq:L_RFS-AA-JPD}
\end{align}
\end{corollary}
\begin{proof}
The result is from Corollary \ref{corollary:delta-GLMB-AA} by using \eqref{eq:LMB-l-EP} and  \eqref{eq:LMB-l-s} and also from Lemma \ref{lemma:glmb-lphd-aa}. \revisNew{It is analogous to} Corollary \ref{corollary:Bernoulli} \revisNew{regardless of the use of the label}. 
We omit the proof detail here. Notably, it is no more required any history information such as $(L,\xi)$ matched among sensors but it is still required that the labels are \revisNew{inter-filter} matched.
\end{proof}
%

Corollary \ref{corollary:LMB-AA} has been given based on the constrained MPD-fit \cite{Gao20GLMB} and further based on the $L_2$ norm metric (disregarding the labels) instead of the KL divergence in \cite[Corollary 3]{Gostar20}. \revisNew{In both of them, the divergence between the MPDs ignores the difference between labels. }

\begin{remark}
A key condition for fusing $\delta$-GLMB as shown in Corollary \ref{corollary:delta-GLMB-AA} is that the history measurement-track association hypotheses $(L,\xi) \subseteq \mathbb{L} \times \Xi$ are well matched among sensors, in addition to the label matching. This involves a process of coordinating \revisNew{and revising perhaps the history estimates 
which} can not be satisfied in real time \cite[Ch.4.4.2]{Da21Recent}. This has been released to only requiring track-set (local hypothesis) matching and label matching in fusing the M-GLMB in Corollary \ref{corollary:M-GLMB-AA} and further to only label-matching in the LMB-AA fusion in Corollary \ref{corollary:LMB-AA}.
\end{remark}

\begin{table*}\footnotesize
	\centering
	\caption{Key variables and functions of different RFS filters for filtering and fusion} \label{tab:KeyEle_RFS}
 \vspace{-1mm}
	\begin{tabular*}{18cm}{ccccccccc}
\hline
		Filters: & Bernoulli	& PHD 	&	CPHD	&	MB	&	MBM	&	($\delta$-)GLBM	& LMB  & M-GLMB \\
\hline
		Items averaged: & $r, s(\mathbf{x})$ &	$D(\mathbf{x})$	& $\rho(n), D(\mathbf{x})$	&  $r_l, s_l(\mathbf{x})$ &	$w_{j,l},r_{j,l}, s_{j,l}(\mathbf{x})$	&	$\omega_i^{(L,\xi)}, s_i^{(\xi)}(\mathbf{x},l)$ 	&	$r_i^{(l)},{s_i}(\mathbf{x}, l)$ &	$\omega_i^{(L)}, s_i^{(L)}(\mathbf{x},l)$	\vspace{1mm} \\
\hline
		(L)PHD-fusion: & Corollary \ref{corollary:Bernoulli}	& Eq. \eqref{eq:PHD-AA-PPP} 	&	Lemma \ref{lemma:IIDC-AA} &	Eq. \eqref{eq:MB-PHA-AA}	& Eq. \eqref{eq:MBM-PHA-AA}	&	Lemma \ref{lemma:glmb-lphd-aa}/Corollary \ref{corollary:delta-GLMB-AA} 	& Corollary \ref{corollary:LMB-AA}  & Corollary \ref{corollary:M-GLMB-AA}   \\
\hline
	\end{tabular*}
\vspace{-3mm}
\end{table*}

\section{Conclusion} \label{sec:conclusion}
We derive suitable AA fusion formulas of different (labeled) RFS filters based on (labeled) PHD consistency, which theoretically explains how the AA fusion can deal with false/missing data and gain better accuracy in target detection and localization, and provides \revisNew{a theoretically unified and exact framework suited for all forms of either RFS or LRFS filters. 
Moreover, the limitation of existing MPD-best-fit derivations for labeled density fusion is pointed out and circumvented in our derivation. The exact AA fusion formula for the general GLMB filter is driven for the first time.} 
The key parameters and functions of different (labeled) RFS filters, which are individually averaged \revisNew{for PHD/LPHD-AA fusion}, are summarized in Table \ref{tab:KeyEle_RFS}.
While most results comply with existing formulas in the literature, our derivation is \revisNew{well-motivated, theoretically unified and exact.} 
\revisNew{It exposes the essence} of existing consensus-driven multi-sensor linear fusion RFS filters which merely seek consensus over the first order moments of the MPDs, namely the (labeled) PHD, rather than the complete posteriors. 
\arxivNew{Thanks to this, two computationally efficient, heterogenous PHD-MB-LMB filter fusion approaches via the coordinate decent method and variational approximation are given in \cite{Li23Heterogeneous} and \cite{Li23HeterVA}, respectively}. 

Open-access Matlab codes for some representative AA-fusion based multi-sensor RFS filters are available 
in the URL: sites.google.com/site/tianchengli85/matlab-codes/aa-fusion. \revisNew{They, as well as some other available PHD-AA implementations for homogeneous RFS/LRFS filter fusion, can be combined for heterogenous RFS filter cooperation
via averaging their respective unlabled/labeled PHDs. That being said, the implementation is not straightforward and may have to resort to sophisticated optimization technologies.} 
Compared with the unlabeled case, the labeled PHD fusion is more challenging due to the need of label matching which may not be able to be solved convincingly.
Scenario-specific label matching strategies need further investigation.

\bibliographystyle{IEEEtran}
\bibliography{RFS_AA}


\end{document}